\documentclass[11pt]{article}

\usepackage[usenames,dvipsnames]{color}

\usepackage{graphicx}

\usepackage{ifthen}
\usepackage{array}
\usepackage{colortbl}
\usepackage{longtable}
\usepackage{amsxtra}	
\usepackage{amscd}	

\usepackage{amsmath}
\usepackage{amssymb}
\usepackage{amsfonts}
\usepackage{amsthm}

\usepackage{mathrsfs}	

\theoremstyle{plain}
\newtheorem{theorem}{Theorem}[section]
\newtheorem{corollary}[theorem]{Corollary}
\newtheorem{proposition}[theorem]{Proposition}
\newtheorem{lemma}[theorem]{Lemma}
\newtheorem{definition}[theorem]{Definition}
\newcommand{\define}[1]{{\em #1}}

\newtheoremstyle{myRemark}%
   {}
   {}
   {}
   {}
   {\bfseries}
   {.}
   {5pt plus 1.0pt minus 1.0pt}
   {}

\theoremstyle{myRemark}
\newtheorem{remark}[theorem]{Remark}

\newcommand{\thmref}[1]{{\upshape\ref{#1}}%
}

\newcounter{proofenumi}
\newenvironment{proofenumerate}%
   {\begin{list}{(\roman{proofenumi})}{\usecounter{proofenumi}%
	    \setlength{\itemindent}{5.475pt}%
	    \setlength{\labelsep}{5.475pt}%
	    \setlength{\leftmargin}{0pt}%
	    }}%
   {\end{list}}
\newcommand{\C}{\mathbb C}
\newcommand{\N}{\mathbb N}
\newcommand{\R}{\mathbb R}
\newcommand{\Z}{\mathbb Z}

\newcommand{\mA}{\mathcal A}
\newcommand{\mB}{\mathcal B}
\newcommand{\mD}{\mathcal D}
\newcommand{\mN}{\mathcal N}
\newcommand{\mT}{\mathcal T}

\newcommand{\FA}{\mathfrak A}	
\newcommand{\FB}{\mathfrak B}	
\newcommand{\FR}{\mathfrak R}	

\newcommand{\fs}{\mathfrak s}	
\newcommand{\ft}{\mathfrak t}	

\renewcommand{\a}{\alpha}

\renewcommand{\phi}{\varphi}
\renewcommand{\theta}{\vartheta}

\newcommand{\eps}{\varepsilon}
\newcommand{\ga}{\gamma}
\newcommand{\la}{\lambda}

\newcommand{\om}{\omega}

\newcommand{\Om}{\Omega}

\newcommand{\specspace}{\mathcal L}

\newcommand{\closed}{\mathscr{C}}	

\newcommand{\spectrum}{\sigma}
\newcommand{\pointspec}{\sigma_p}
\newcommand{\essspec}{\sigma_{ess}}
\newcommand{\disspec}{\sigma_{d}}
\newcommand{\lae}{\la_e}
\newcommand{\dimlae}{n}			

\newcommand{\hil}{\mathcal H}
\newcommand{\LL}{\mathscr{L}}	
\newcommand{\Ltwo}[2]{\LL^2(#1, #2)}  

\providecommand{\range}{\mathop{\rm ran}\nolimits}
\providecommand{\codim}{\mathop{\rm codim}\nolimits}
\providecommand{\re}{\mathop{\rm Re}\nolimits}		

\providecommand{\sign}{\mathop{\rm sign}\nolimits}

\newcommand{\id}{I}

\newcommand{\diff}[1]{\frac{\mathrm{d}}{\mathrm{d} #1}}
\newcommand{\pa}[1]{\frac{\partial}{\partial #1}}

\newcommand{\rd}{{\rm d}}

\newcommand{\SC}[1]{S_{#1}}
\newcommand{\SCmin}[1]{S_{#1}^{\mathrm{{\scriptscriptstyle[}min{\scriptscriptstyle]}} }}

\newcommand{\tform}[1]{\ft_{#1}}
\newcommand{\tformclosed}[1]{\widetilde\ft_{#1}}

\newcommand{\SCfs}[1]{\fs_{#1}}
\newcommand{\SCft}[1]{\ft_{#1}}

\newcommand{\sig}[1]{\sigma_{#1}}

\newcommand{\Dup}{c_2}
\newcommand{\Alo}{c_1}
\newcommand{\Dcup}{c_2}
\newcommand{\Dclo}{c_2^-}
\newcommand{\Acup}{c_1^+}
\newcommand{\Aclo}{c_1}

\newcommand{\laup}{\lambda^{[u]}}		
\newcommand{\lalo}{\lambda^{[l]}}		

\newcommand{\lanumplus}[1][noargument]{
\ifthenelse{\equal{#1}{noargument}}{\la^{[num]}}{\la^{[num]}_{#1}}%
}
\newcommand{\lanumminus}[1][noargument]{
\ifthenelse{\equal{#1}{noargument}}{\la^{[num]}}{\la^{[num]}_{-#1}}%
}

\newcommand{\lanumplusS}[1][noargument]{
\ifthenelse{\equal{#1}{noargument}}{\la^{[num]}_S}{\la^{[num]}_{S,#1}}%
}
\newcommand{\lanumminusS}[1][noargument]{
\ifthenelse{\equal{#1}{noargument}}{\la^{[num]}_S}{\la^{[num]}_{S,-#1}}%
}
\newcommand{\lanumS}[1][noargument]{
\ifthenelse{\equal{#1}{noargument}}{\la^{[num]}_S}{\la^{[num]}_{S,#1}}%
}

\newcommand{\laQ}{\la_Q}

\newcommand{\ax}{\alpha}	
\newcommand{\abx}{\alpha_{21}}	

\newcommand{\eignumber}{n}			
\newcommand{\Aeignumber}{n}			

\definecolor{SuffernColour}{gray}{0.9}

\arrayrulecolor{black}

\newcolumntype{N}{>{\columncolor{NumColour}}r}
\newcolumntype{S}{>{\columncolor{SuffernColour}}r}

\newcommand{\tm}[4]
   {\left(\begin{smallmatrix}#1&#2\\#3 &#4\end{smallmatrix}\right)}

\newcommand{\ttv}[2]{(#1,\,#2)^t}	

\newcommand{\tv}[2]
   {\left(\begin{smallmatrix}#1\\#2\end{smallmatrix}\right)}

\newlength{\arraycolsepSave}	
\newcommand{\mv}[1]{
   \setlength{\arraycolsepSave}{\arraycolsep}
   \setlength{\arraycolsep}{2pt}
   \left(\begin{array}{c}#1\end{array}\right)
}

\newcommand{\scalar}[2]{(#1,\, #2)}
\newcommand{\bigscalar}[2]{\bigl(#1,\, #2\bigr)}
\newcommand{\norm}[1]{\|#1\|}
\newcommand{\form}[2]{[#1,\, #2]}

\newcommand{\ltwo}[1]{\|#1\|_2}			

\newcommand{\msub}[2]{\genfrac{}{}{0pt}{}{#1}{#2}}

\providecommand{\maxp}{\mathop{\rm max\vphantom{p}} }
\providecommand{\minp}{\mathop{\rm min\vphantom{p}} }

\newcommand{\noz}[1]{#1^\times}

\newcommand{\ts}{\textstyle}

\makeatletter
\renewcommand*\refstepcounter[1]{\stepcounter{#1}%
  \protected@edef\@currentlabel{%
    \csname p@#1\expandafter\endcsname
      \csname the#1\endcsname
  }%
}
\makeatother

\newcounter{listcounter}
\makeatletter
\newcommand{\CondItem}[1]{
   \renewcommand{\p@listcounter}[1]{#1}	
   \item[#1]
   \refstepcounter{listcounter}		
}
\makeatother

\newenvironment{condlist}
{
\begin{list}{counter is \arabic{listcounter}}{
      \usecounter{listcounter}
      \setlength{\rightmargin}{0.0\textwidth}
      \setlength{\leftmargin}{0.12\textwidth}
      \setlength{\itemindent}{0pt}
      \setlength{\labelsep}{3ex}
      \settowidth{\labelwidth}{ (D2.b) }
   } 
} 
{\end{list}}
\newcommand{\AOne}{A1}
\newcommand{\BOne}{B1}
\newcommand{\BTwo}{B2}
\newcommand{\DOne}{D1}
\newcommand{\DTwoa}{D2}
\newcommand{\DTwob}{D2$'$}
\newcommand{\ATwo}{A2}

\newcommand{\mytitle}{A Variational Principle for Block Operator Matrices and its Application to the Angular Part of the Dirac Operator in Curved Spacetime}

\usepackage[pdfcenterwindow, a4paper, colorlinks=true]{hyperref}
 \hypersetup{
   pdftitle = {\mytitle},
   pdfauthor = {Monika Winklmeier},
   anchorcolor = red,
   linkcolor = Maroon,
   citecolor = OliveGreen,
   urlcolor = Sepia,
} 

\definecolor{refkey}{named}{CadetBlue}
\definecolor{labelkey}{named}{Mahogany}

\begin{document}

\title{\mytitle}
\author{Monika Winklmeier\thanks{Mathematisches Institut, Universit\"at Bern, Sidlerstrasse 5, CH- 3012 Bern,\newline winklmeier@math.unibe.ch}}

\maketitle


\begin{abstract} 
   \noindent
   The operator associated to the angular part of the Dirac equation in the Kerr-Newman background metric is a block operator matrix with bounded diagonal and unbounded off-diagonal entries. 
   The aim of this paper is to establish a variational principle for block operator matrices of this type and to derive thereof upper and lower bounds for the angular operator mentioned above.
   In the last section, these analytic bounds are compared to numerical values from the literature.
\end{abstract} 

\section{Introduction} 
\label{sec:intro}
Variational principles are an important tool to give estimates for eigenvalues of a selfadjoint operator $A$ since no knowledge about the corresponding eigenvectors is needed.
The classical variational principle based on the Rayleigh functional $\frac{\scalar{\psi}{A\psi}}{\|\psi\|^2}$, $\psi\in\mD(A)\setminus\{0\}$, applies only to semibounded operators, see, e.g., \cite{reed_simonIV}.

For example, the eigenvalues of $A$ below its essential spectrum are given by
\begin{align}
   \label{eq:classicalVP}
   \la_n\, =\, 
   \min_{\msub{L\subseteq\mathcal D(A)}{\dim L = n}}
   \max_{\psi\in L\setminus\{0\}}
   \frac{\scalar{\psi}{A\psi}}{\|\psi\|^2}.
\end{align}

In this paper, however, we are interested in analytic bounds for the eigenvalues of the angular part of the Dirac operator in the Kerr-Newman metric which has been studied for instance in~\cite{BSW}.
It is neither bounded from below nor from above as is usually the case for Dirac operators.
Hence, the classical variational principle cannot be applied.

Recently, the Rayleigh functional was used to establish variational principles also for Dirac operators in flat space time.
The idea is to decompose the given Hilbert space on which the Dirac operator is defined into the direct sum of two Hilbert spaces and then, in~\eqref{eq:classicalVP}, to take the minimum over certain subspaces $L$ in the first Hilbert space, and the maximum over certain $\psi$ whose first component lies in $L$.

Griesemer and Siedentop~\cite{GS} have proved a variational principle for the eigenvalues of a block operator matrix $\mT$ on a Hilbert space $\hil=\hil_1\oplus\hil_2$ in a gap of its essential spectrum where the authors did not assume that the operator is semibounded. 
They assumed that the spectra of the operators on the diagonal do not overlap, so that, roughly speaking, the decomposition of the Hilbert space $\hil$ into spectral subspaces of $\mT$ is close to the given decomposition. 
In \cite{GLS}, Griesemer, Lewis and Siedentop gave a variational principle for the Dirac operator with Coulomb potential.
In the case of the block operator matrix in section~\ref{sec:application}, however, the spectra of the diagonal entries do overlap, hence this principle cannot be applied.
Under assumptions different from those in~\cite{GS}, Dolbeault, Esteban and S{\'e}r{\'e} proved a variational principle for the eigenvalues of operators in gaps \cite{DES00a}, in particular they considered Dirac operators with Coulomb potential~\cite{DES00b}.
Their techniques differ from those used in \cite{GS} and \cite{GLS}, but they also make use of the Rayleigh functional.

Various types of block operator matrices and their spectral properties have been investigated recently. 
A survey of some recent results can be found in~\cite{Tr99}. 
In this paper we are interested in so-called off-diagonally dominant selfadjoint block operator matrices
\begin{align}
   \mT\, =\, 
   \begin{pmatrix} T_{11} & T_{12} \\ T_{12}^* & T_{22} \end{pmatrix},
   \qquad\qquad
   \mD(\mT) = \mD(T_{12}^*)\oplus\mD(T_{12})
\end{align}
on a Hilbert space $\hil=\hil_1\oplus\hil_2$ where the linear operators $T_{ij}(\hil_j\rightarrow\hil_i)$ are closed and $T_{11}$ and $T_{22}$ are symmetric and bounded with respect to $T_{12}^*$ and $T_{12}$ respectively.
Further we assume that the operator $T_{11}$ is bounded from below by $\Alo$ and that $T_{22}$ is bounded from above by $\Dup$.
However, we do not assume that the spectra of $T_{11}$ and $T_{22}$ are disjoint.
The crucial step to obtain a variational principle for the eigenvalues of $\mT$ is to associate with it an operator valued function, the so-called Schur complement
\begin{align*}
   \SC{1}(\la)\, =\, T_{11}-\la - T_{12}(T_{22}-\la)^{-1}T_{21},
   \qquad \la\in\rho(T_{22}).
\end{align*}

The Schur complement plays an important role in the Schur-Frobenius factorisation of block operator matrices, see, e.g., \cite{Nagel89}, \cite{ALMS}, \cite{ALMSau}.
It turns out that $\pointspec(\mT)\cap(\Dup,\lae) = \pointspec(\SC{1})\cap(\Dup,\lae)$ where $\lae = \inf \essspec(\SC{1})$.
To the Schur complement $\SC{1}$ we can apply a variational principle proved by  by Binding, Eschw{\'e} and Langer~\cite{BEL00} and by Eschw{\'e} and Langer~\cite{eschwe}.
We follow an approach that has already been used by Langer, Langer and Tretter in~\cite{LLT02} for the so-called diagonally dominant case with bounded off-diagonal elements $T_{12}$ and $T_{12}^*$.
The off-diagonally dominant case has been studied also by Kraus, Langer and Tretter~\cite{KLT04} under the assumption the diagonal entries $T_{11}$ and $T_{22}$ of the block operator matrix are bounded.

The main result of this paper is a variational characterisation of the eigenvalues of off-diagonally dominant block operator matrices $\mT$ where one of the diagonal entries, e.g. $T_{11}$, of the block operator matrix may be unbounded (Theorem~\thmref{theorem:schurvar:var2}).
Under additional assumptions on the operator $T_{12}$ we derive in Theorem~\ref{theorem:schurvar:estimate} the following upper and lower bounds for the eigenvalues $\la_1\ge\la_2\ge \dots $ of $\mT$ which are greater than $\Dup$:
\begin{alignat*}{2}
   \la_{\eignumber}\ &\le\
   \frac{\abx}{2} \sqrt{\nu_{\eignumber+n_0}}  + 
   \sqrt{ \nu_{\eignumber+n_0} + {\ts\frac{1}{4}} (\abx\sqrt{\nu_{\eignumber+n_0}}+\|T_{22}\|+\ax)^2  }
   \ + \ts\frac{1}{2} \left( \ax + \Dup  \right)
   ,\quad
   \\[2ex]
   \la_{\eignumber}\ 
   &\ge\ \sqrt{\nu_{\eignumber+n_0}}\  + \ts\frac{1}{2}(\Alo -\|T_{22}\|)
\end{alignat*}
where the constants $\abx$, $\ax$, $\Aclo$ and $\Dup$ are such that 
\begin{gather*}
   \begin{alignedat}{3}
      \scalar{y}{T_{11}y}\ &\le\ \Aclo\, \|y\|^2,
      \quad
      & y\in\hil_2,
      \qquad
      \scalar{x}{T_{22}x}\ &\le &\ \Dup\, \|x\|^2,
      \quad
      &x\in\hil_1,
   \end{alignedat}
   \\[1ex]
   \|T_{11}x\|\ \le\ \ax \|x\| + \abx \|T_{12}^*x\|, 
   \qquad x\in\mD(T_{12}^*),
\end{gather*}
and $0\, <\, \nu_1\, \le\, \nu_2\, \le\, \dots $ are the eigenvalues of $T_{12}^{}T_{12}^*$.
The index shift $n_0$ is given by $n_0=\min\limits_{\la>\Dcup}\dim\specspace_{(-\infty,0)} \SC{1}(\la)$ where $\specspace_{(-\infty,0)} \SC{1}(\la)$ denotes the spectral subspace of $\SC{1}(\la)$ corresponding to the interval $(-\infty, 0)$.
In the case $T_{11}=0$ and $T_{22}=0$ the formulae above give the exact values of eigenvalues, hence Theorem~\ref{theorem:schurvar:var3} can be regarded as a perturbational result for block operator matrices with a certain type of unbounded perturbation.

The variational principle established in this paper is applied to the angular part of the Dirac operator in curved spacetime to derive upper and lower bounds for its eigenvalues in Theorem~\thmref{theorem:angschur:var1}.
These bounds are given explicitly in terms of the physical quantities involved.
Suffern, Fackerell and Cosgrove~\cite{SFC83} derived numerical approximations for the eigenvalues by applying a power series ansatz in two of the physical variables involved.
They obtained numerical approximations for the eigenvalues, but they did not give asymptotics of the eigenvalues as are obtained by our analytical approach.

\par
The paper is organised as follows.
In section~\ref{sec:SC} we consider off-diagonally dominant block operator matrices and define the Schur complements associated with them. 
%
 Section~\ref{sec:variatonal} contains the variational principle for block operator matrices which is applied to the angular part of the Dirac operator in curved spacetime in section~\ref{sec:application}.
Finally, in section~\ref{subsec:comparison} the analytic bounds are compared to numerical approximations calculated by Suffern et al.~\cite{SFC83}.
%
%


\section{Schur complements and sesquilinear forms related to a block operator matrix} 
\label{sec:SC}
Let $\hil_1$ and $\hil_2$ be Hilbert spaces with norm and scalar product denoted by $\norm{\cdot}_j$ and $\scalar{\,\cdot\ }{\cdot\,}_j$, $j=1,2$.
Consider the Hilbert space $\hil := \hil_1\oplus\hil_2$ equipped with the norm $\norm{\cdot}$ and the scalar product $\scalar{\phi}{\psi} = \scalar{\phi_1}{\phi_1}_1 + \scalar{\phi_2}{\psi_2}_2$ for $\phi=\ttv{\phi_1}{\phi_2},\ \psi=\ttv{\psi_1}{\psi_2}\in\hil$.
On $\hil$ we consider block operator matrices
\begin{align*}
   \mT\, =\, 
   \begin{pmatrix} T_{11} & T_{12} \\ T_{21} & T_{22} \end{pmatrix},
   \hspace{5ex}
   \mD(\mT)\, =\, \mD_1\oplus\mD_2\, \subseteq\, \hil_1\oplus\hil_2
\end{align*}
with linear operators $T_{ij}(\hil_j\rightarrow\hil_i)$, $i,j=1,2$.

In this paper we will always assume that the following conditions concerning the entries of $\mT$ hold:

\begin{condlist}
\CondItem{{\rm(\BOne)}}
   \label{B1}
   {$T_{12}$ is a closed densely defined operator from $\hil_2$ to $\hil_1$ and $T_{12}^*=T_{21}$;}
   \CondItem{{\rm (\AOne)}}
   \label{A1}
   $\mD(T_{12}^*)\subseteq\mD(T_{11})$ and 
   $T_{11}$ is symmetric in $\hil_1$ and semibounded from below, i.e., there is a  constant
   $\Alo\in\R$ such that 
   \begin{align*}
      \scalar{x}{T_{11}x}\ \ge\ \Alo\|x\|^2,\qquad x\in\mD(T_{11});
   \end{align*}
   \CondItem{{\rm(\DOne)}}%
   \label{D1}%
   $\mD(T_{12})\subseteq\mD(T_{22})$ and $T_{22}$ is symmetric in $\hil_2$ and 
      semibounded from above, i.e., there is a constant $\Dup\in\R$ such 
      that  
      \begin{align*}
	 \scalar{x}{T_{22}x}\ \le\ \Dup\|x\|^2, \qquad x\in\mD(T_{22});
      \end{align*}
      furthermore, $T_{22}$ is closed and $(\Dup, \infty)\subseteq \rho(T_{22})$.
\end{condlist}
We always assume that the block operator matrix $\mT$ is given by
\begin{condlist}
   \CondItem{{\rm($\mT 1$)}}
   \label{mT1}
   $\mT=\begin{pmatrix}T_{11}&T_{12}\\ T_{12}^*&T_{22}\end{pmatrix}$,
      \quad $\mD(\mT) = \mD(T_{12}^*)\oplus\mD(T_{12})$.
\end{condlist}

\begin{remark}
   \begin{enumerate}

      \item 
      Since $T_{11}$ is closable by assumption, the condition concerning its domain implies (see~\cite[chap. IV, remark 1.5]{kato}) that $T_{11}$ is $T_{21}$-bounded, i.e., that there are positive numbers 
      $\ax$  and $\abx$ such that 
      \begin{align*} 
	 \|T_{11}x\|\ \le\ \ax\,\|x\| + \abx\,\|T_{21}x\|,
	 \qquad x\in\mD(T_{21}).
      \end{align*} 

      \item 
      Condition \thmref{D1} implies that $T_{22}$ is even selfadjoint because
      the \index{defect index} defect index of the closed operator $T_{22}$ is constant on the connected set $\C\setminus\overline{W(T_{22})}$ where $W(T_{22})= \{ \scalar{x}{T_{22}x}\, :\, x\in\mD(T_{22}), \|x\|=1 \}$ is the numerical range of $T_{22}$.
      Now, $\rho(T_{22})\cap \C \setminus\overline{W(T_{22})}$ being nonempty implies that $T_{22}$ has zero defect, hence it is essentially selfadjoint. 
      Since $T_{22}$ is already closed, its selfadjointness is proved.
   \end{enumerate}
\end{remark}

Observe that the above conditions do not imply that $\mT$ is closed. 
\medskip\par

Next we associate an operator valued function, the so-called Schur complement, to the block operator matrix $\mT$.
In 
Corollary~\thmref{cor:schurvar:essspec}
we show that the spectrum of the Schur complement and the spectrum of $\mT$ are related.
For $\la\in\rho(T_{22})$ define 
\begin{align}\label{eq:SCmin}
   \begin{aligned}
      \mD(\SCmin{1}(\la))\, &:=\, \{ \psi\in\mD(T_{12}^*)\, :\,
	 (T_{22}-\la)^{-1}T_{12}^*\psi \in\mD(T_{12}) \},\\[1ex]
      \SCmin{1}(\la)\, &:=\,
      T_{11}-\la - T_{12}(T_{22}-\la)^{-1}T_{12}^*.
   \end{aligned}
\end{align}
The family $\SCmin{1}(\la)$, $\la\in\rho(T_{22})$ is the so-called \define{minimal Schur complement of $\mT$}.
In the following we are interested in real $\la$.

\begin{proposition}\label{prop:schurvar:schur}
   Assume that conditions \ref{B1}, \ref{A1}, \ref{D1} and \ref{mT1} hold.
   Then  for all $\la> \Dup$ the operator $\SCmin{1}(\la)$ is bounded from below. 
   If in addition the condition 
   \begin{condlist}
      \CondItem{{\rm(\DTwoa)}}{
	 $T_{22}$ is bounded}
      \label{D2a}
   \end{condlist}
   holds, then $\SCmin{1}(\la)$ is also symmetric, and therefore densely defined and closable.
\end{proposition}
\begin{proof}
   Because of the inclusion $\mD(T_{12}^*)\subseteq\mD(T_{11})$, the Schur complement is well defined. 
   To show that $\SCmin{1}(\la)$ is semibounded, we use that $(\la-T_{22})^{-1}$ is a positive operator for $\la>\Dup$; hence, for all $x\in\mD(\SCmin{1}(\la))$ we have
   \begin{align*} 
      \scalar{x}{\SCmin{1}(\la)x}\ 
      &=\
      \scalar{x}{(T_{11}-\la)x} - \scalar{x}{T_{12}(T_{22}-\la)^{-1}T_{12}^*x}\\ 
      &=\ \scalar{x}{(T_{11}-\la)x} + \scalar{T_{12}^*x}{(\la-T_{22})^{-1}T_{12}^*x}\ 
      \ge\, (\Alo-\la)\|x\|^2.
   \end{align*} 
   In particular it follows that the scalar product on the left hand side is real, hence $\SCmin{1}(\la)$ is formally symmetric. 
   It remains to be shown that $\mD(\SCmin{1}(\la))$ is dense in $\hil_1$.
   By assumption, the operator $(\la-T_{22})^{-1}$ is selfadjoint, bounded and positive for fixed $\la>\Dup$. 
   Hence there exists a positive square root $(\la-T_{22})^{-\frac{1}{2}}$ which is also bounded and selfadjoint. 
   Therefore $((\la-T_{22})^{-\frac{1}{2}}T_{12}^*)^*=T_{12}(\la-T_{22})^{-\frac{1}{2}}$ holds. 
   Condition \ref{D2a} implies that the operator $(\la-T_{22})^{-\frac{1}{2}}T_{12}^*$ is closed, hence by  von~Neumann's theorem (see, for instance, \cite[chap. V, theorem~3.24]{kato}) the operator 
   $((\la-T_{22})^{-\frac{1}{2}}T_{12}^*)^* ((\la-T_{22})^{-\frac{1}{2}}T_{12}^*) = -T_{12}(T_{22}-\la)^{-1}T_{12}^* $ with domain 
   \begin{multline*}
      \{x\in\mD( (\la-T_{22})^{-\frac{1}{2}}T_{12}^*)\ :\ 
	 (\la-T_{22})^{-\frac{1}{2}}T_{12}^*x\in\mD( T_{12}(\la-T_{22})^{-\frac{1}{2}} ) \} \\
      = \{ x\in\mD(T_{12}^*)\, :\, (T_{22}-\la)^{-1}T_{12}^*\,x\,\in\mD(T_{12}) \}\
      =\ \mD(\SCmin{1}(\la))
   \end{multline*}
   is selfadjoint and its domain is a core of $(T_{22}-\la)^{-\frac{1}{2}}T_{12}^*$;
   in particular, its domain is dense in $\hil_1$.
\end{proof}

In order to find selfadjoint extensions of the operators $\SCmin{1}(\la)$ we define sesqui\-linear forms associated with them.

\begin{proposition}\label{prop:schurvar:form}
   Assume that the conditions \ref{B1}, \ref{A1}, \ref{D1}, \ref{D2a} and \ref{mT1} hold. 
   Then for all $\la\in(\Dup,\infty)$ the sesquilinear form 
   \label{first:SCfsmin}
   \begin{align*}
      \mD(\SCfs{1}(\la)) := \mD(T_{12}^*),
      \;\,
      \SCfs{1}(\la)[u,\, v] := \bigscalar{u}{(T_{11}-\la)v} - \bigscalar{T_{12}^*u}{(T_{22}-\la)^{-1}T_{12}^*v} 
   \end{align*}
   in $\hil_1$ is symmetric, semibounded from below and closed.
\end{proposition}
\begin{proof}
   The symmetry and boundedness from below can be shown as in Proposition~\thmref{prop:schurvar:schur}. 
   Since $T_{22}$ is bounded by assumption, it follows that also $(T_{22}-\la)^{\frac{1}{2}}$ is bounded and therefore the  operator $(\la-T_{22})^{-\frac{1}{2}}T_{12}^*$ is closed.
   Consider the auxiliary sesquilinear form 
   \begin{align*}
      \mD(\SCft{12}(\la))\ :=\ \mD(T_{12}^*),
      \qquad
      \SCft{12}(\la)[u,\, v]\ :=\ \bigscalar{T_{12}^*u}{(\la - T_{22})^{-1}T_{12}^*v} .
   \end{align*}
   For every $\tform{12}(\la)$-convergent sequence $(x_n)_{n\in\N}\subseteq\mD(T_{12}^*)$ we have
   \begin{align*}
      \|(T_{22}-\la)^{-\frac{1}{2}}T_{12}^*(x_n-x_m)\|\ =\
      \tform{12}(\la)[x_n-x_m]\ \rightarrow 0,
      \qquad\qquad n,\, m\rightarrow\infty.
   \end{align*}
   Since $(T_{22}-\la)^{-\frac{1}{2}}T_{12}^*$ is closed, it follows that
   $x:=\lim\limits_{n\rightarrow\infty}x_n \in\mD((T_{22}-\la)^{-\frac{1}{2}}T_{12}^*) = \mD(T_{12}^*)$.
   This shows that $\tform{12}(\la)$ is closed. 

   The operator $T_{11}-\la$ is symmetric and bounded from below, hence it is form-closable, i.e., the symmetric form $\tform{11}(\la)$ defined by $\tform{11}(\la)\form{\phi}{\psi} = \scalar{T_{11}\phi}{\psi}$ for $\phi,\ \psi\in\mD(T_{11})$ is closable;
    let $\tformclosed{11}(\la)$ denote its closure.
   Then it follows that the form 
   $\SCfs{1}(\la) = \tformclosed{11}(\la) + \tform{12}(\la)$ with domain $\mD(\tformclosed{11}(\la))\cap\mD(\tform{12}(\la))$ is also closed. 
   Since $\mD(\tform{12}(\la))\subseteq\mD(\tform{11}(\la))\subseteq\mD(\tformclosed{11}(\la))$, the form $\SCfs{1}(\la)$ with domain $\mD(\tform{12}(\la)) = \mD(T_{12}^*)$ is closed.
\end{proof}

\begin{remark}
   If we assume that instead of \ref{D2a} the condition
   \begin{condlist}
      \CondItem{{\rm(\DTwob)}}{
	 $(T_{22}-\la)^{-1}(\range( T_{12}^*\cap \mD(T_{12}) )\ 
	    \subseteq\ \mD(T_{12})$,\quad $\la>\Dup$.
      }
      \label{D2b}
   \end{condlist}
   holds, then Proposition~\thmref{prop:schurvar:schur} is also valid. 
   However, the sesquilinear form $\SCfs{1}$ defined in Proposition~\thmref{prop:schurvar:form}
   is closable but not necessarily closed (cf.~\cite{thesis}). 
\end{remark}

Throughout the rest of this section we assume that condition \ref{D2a} holds.
Since the forms $\SCfs{1}(\la)$ are closed, symmetric and semibounded, there are uniquely defined selfadjoint operators $\SC{1}(\la)$ such that 
$\mD(\SC{1}(\la)) \subseteq \mD(\SCfs{1}(\la))$ and 
\begin{align*}
   \SCfs{1}(\la)\form{\phi}{\psi}\, =\, \scalar{S_1(\la)\phi}{\psi},
   \hspace{2ex}
   \phi\in\mD(\SC{1}(\la)),\ \psi\in\mD(\SCfs{1}(\la)).
\end{align*}
In addition, if for fixed $\psi\in\mD(\SCfs{1}(\la))$ there exists a $w\in\hil_1$ such that $\SCfs{1}(\la)\form{\psi}{\phi} = \scalar{w}{\phi}$ for all $\phi$ belonging to a core of $\SCfs{1}(\la)$, then $\psi\in\mD(\SC{1}(\la))$ and $\SC{1}(\la)\psi = w$.

Obviously, for all $\la\in\rho(T_{22})$ the operator $\SC{1}(\la)$ is an extension of $\SCmin{1}(\la)$; it is the so-called \define{Friedrichs extension}.
Note that in general this extension is not the only possible selfadjoint extension (cf.~\cite[chap. VI]{kato}).
We call the family $\SC{1}(\la)$, $\la\in\rho(T_{22})$, the \define{Schur complement} of the block operator matrix $\mT$.

\par
The following proposition gives a sufficient condition such that $\SCmin{1}(\la)$ is already selfadjoint.

\begin{proposition}\label{prop:schurvar:inclusion}
   Assume that the conditions \ref{B1}, \ref{A1}, \ref{D1}, \ref{D2a} and \ref{mT1} hold and let $\la\in(\Dup,\infty)$.
   In addition assume that condition
   \begin{condlist}
      \CondItem{{\rm(\ATwo)}}
      \label{A2}
      $T_{11}$ is symmetric and $T_{12}(T_{22}-\la)^{-1}T_{12}^*$-bounded with relative bound less than $1$, i.e., there are real numbers $\a>0$ and $1>\widetilde\a \ge 0$ {\upshape(}which may depend on $\la${\upshape )} such that for all $x\in\mD(T_{12}(T_{22}-\la)^{-1}T_{12}^*)$
 	 \begin{align*}
	    \|T_{11}x\|\, \le\, \a\, \|x\| + \widetilde\a\, \|T_{12}(T_{22}-\la)^{-1}T_{12}^*x\|, 
 	 \end{align*}
   \end{condlist}
   holds.
   Then 
   $\SC{1}(\la)=\SCmin{1}(\la)$; in particular, $\SCmin{1}(\la)$ is selfadjoint.
\end{proposition}
\begin{proof}
   Recall that $T_{12}(T_{22}-\la)^{-1}T_{12}^*$ is selfadjoint, see proof of Proposition~\thmref{prop:schurvar:schur}. 
   If $T_{22}$ is bounded and $T_{11}$ is bounded with respect to $T_{12}(T_{22}-\la)^{-1}T_{12}^*$ with relative bound less than $1$, then the Kato-Rellich theorem~\cite[chap. V, theorem~4.3]{kato} yields that $\SCmin{1}(\la)$ is selfadjoint.
   Hence $\SCmin{1}(\la)=\SC{1}(\la)$ follows.
\end{proof}

Next we will show the relation between the spectra of $\SC{1}(\la)$ and $\mT$.

\begin{proposition}\label{prop:schurvar:spectrum}
   For given Hilbert spaces $\hil_1$ and $\hil_2$ we consider linear operators $T_{ij}(\hil_j\rightarrow\hil_i)$, $i,j = 1,2$, with $\mD(T_{21})\subseteq\mD(T_{11})$ and $\mD(T_{12})\subseteq\mD(T_{22})$. 
   Let $\mT=\tm{T_{11}}{T_{12}}{T_{21}}{T_{22}}$ be the 
   block operator matrix with domain $\mD(\mT) = \mD(T_{21})\oplus\mD(T_{12})$ in the Hilbert space $\hil_1\oplus\hil_2$.
   If\, $T_{22}$ is bijective, then the operator 
   \begin{align*}
      S\ :=\ T_{11} - T_{12}T_{22}^{-1}T_{21},\
      \qquad\mD(S):= \{x\in\mD(T_{21})\, :\, T_{22}^{-1}T_{21}x\in\mD(T_{12})\}
   \end{align*}
   is well defined and the following holds:
   \begin{enumerate}
      \item\label{item:schurvar:kernel}
      \hfill
      $\mT \text{ is injective }
      \qquad\Longleftrightarrow\qquad
      S\text{ is injective}$
      \hspace*{\fill}\\
      and
      \begin{align*}
	 \ker(\mT)\ &=\ 
	 \left\{\mv{f \\ -T_{22}^{-1}T_{21}f}\ :\ f\in\ker S \right\}\
	 \cong\ \ker(S).
      \end{align*}
   \end{enumerate}
   \begin{enumerate}
      \setcounter{enumi}{1}
      \item\label{item:schurvar:surj} 
      If additionally $T_{21}$ is surjective, then 
      $\range(S)\oplus\{0\} = \range(\mT)\cap (\hil_1\oplus\{0\})$ and
      \\[2ex]
      \hspace*{\fill}
      $\mT \text{ is surjective}
      \qquad\Longleftrightarrow\qquad
      S\text{ is surjective}$.
      \hspace*{\fill}
   \end{enumerate}
\end{proposition}
\begin{proof} 
   \begin{proofenumerate}
      \item
      First assume that $\mT$ is not injective. 
      Then there are $f\in\mD(T_{21})$, $g\in\mD(T_{12})$ such that
      \begin{align*}
	 T_{11}f + T_{12}g = 0,\qquad
	 T_{21}f + T_{22}g = 0
	 \qquad\text{and}\qquad
	 \mv{f\\ g}\neq 0.
      \end{align*}
      From the second equality it follows that 
      $T_{22}^{-1}T_{21}f = -g\in\mD(T_{12})$. 
      Consequently, $f$ lies in $\mD(S)$ and $f\neq 0$. 
      Inserting the expression for $g$ into the first equality gives $Sf =0$, hence $S$ is not injective.
      Now assume that $S$ is not injective and fix an element $f\neq 0$ in its kernel.
      For $g:=-T_{22}^{-1}T_{21}f$ it follows that
      \begin{align*}
	 0\ &=\ Sf\ =\ 
	 T_{11}f - T_{12}T_{22}^{-1}T_{21}f\
	 =\ T_{11}f + T_{12}g,\\
	 0\ &=\ g+ T_{22}^{-1}T_{21}f\ =\ 
	 T_{22}^{-1}(T_{22}g + T_{21}f).
      \end{align*}
      Since $T_{22}^{-1}$ is injective, the above equations show 
      $0\neq\ttv{f}{g}\in \ker(\mT)$. 
      
      \item 
      For every $f\in\mD(S)$, the element $g:= -T_{22}^{-1}T_{21}f$ lies in $\mD(T_{12})$. 
      Consequently, $\ttv{f}{g}\in\mD(\mT)$ and 
      \begin{align*}
	 \mT\mv{f\\g}\ =\ \mv{T_{11}f + T_{12} g \\ T_{21}f + T_{22} g}\
	 =\ \mv{T_{11}f - T_{12}T_{22}^{-1}T_{21}f \\ 0}\
	 =\ \mv{Sf\\ 0}
      \end{align*}
      which implies that $\range(S)\oplus\{0\}\subseteq\range(\mT)\cap(\hil_1\oplus\{0\})$.
      Conversely, let $\ttv{f}{g}\in\mD(\mT)$ such that 
      $\mT\ttv{f}{g}=  \ttv{x}{0}$ for some $x\in\hil_1$. 
      From $T_{21} f + T_{22} g = 0$ it follows that $g=-T_{22}^{-1}T_{21} f \in\mD(T_{12})$. 
      Thus we have $f\in\mD(S)$ and 
      \begin{align*}
	 x\ =\ T_{11} f + T_{12} g\ =\ T_{11}f - T_{12}T_{22}^{-1}T_{21} f\ 
	 =\ Sf,
      \end{align*}
      implying $\range(\mT)\cap (\hil_1\oplus\{0\})\subseteq \range(S)\oplus\{0\}$.
      In particular, the surjectivity of $\mT$ implies that of $S$.
      Finally, assume that $S$ is surjective and fix $\ttv{x}{y}\in\hil_1\oplus\hil_2$. 
      Since $\range(T_{21})=\hil_2$ by assumption, there is an $f'\in\mD(T_{21})\subseteq\mD(T_{11})$ such that $T_{21}f'=y$.
      Therefore, $\ttv{f'}{0}$ lies in the domain of $\mT$ and we have $\mT\ttv{f'}{0}=\ttv{T_{11}f'}{y}$. 
      Since we have already shown that $\range(S)\oplus\{0\}=\range(\mT)\cap(\hil_1\oplus\{0\})$, the surjectivity of $S$ implies 
      $\hil_1\oplus\{0\} = \range(\mT)\cap (\hil_1\oplus\{0\}) \subseteq\range(\mT)$,
      hence we finally have 
      $\ttv{\vphantom{f}x}{y} = \mT\ttv{f'}{0} + \ttv{x-T_{11}f'}{0}\in \range{\mT}$ because both terms on the right hand side lie in $\range(\mT)$.
      \qedhere
   \end{proofenumerate}
\end{proof}

\par\medskip

The spectrum and resolvent set of an operator valued function are defined as follows.
\begin{definition}\label{definition:schurvar:spectrum}
   Let $S=(S(\zeta))_\zeta$ be a family of closed operators, where $\zeta$ varies in some set $U\subseteq \C$.  Then the \define{spectrum\index{spectrum of an operator function}}, \define{point spectrum} and \define{resolvent set\index{resolvent set of an operator function}} of $S$ are defined as
   \begin{align*}
      \spectrum(S)\ &:=\  \{\zeta\in U\, :\, 0\in \spectrum(S(\zeta))\},\\
      \pointspec(S)\ &:=\  \{\zeta\in U\, :\, 0\in \pointspec(S(\zeta))\},\\
      \rho(S)\ &:=\  \{\zeta\in U\, :\, 0\in \rho(S(\zeta))\}.
   \end{align*}
   Analogous definitions apply to the other parts of the spectrum of $S$, e.g., the \define{essential spectrum}.
\end{definition}

Recall  that for a linear operator $S$ the \define{\index{spectrum!essential}essential spectrum} and \define{\index{spectrum!discrete}discrete spectrum} are defined by
\begin{align*}
   \essspec(S)\, &:=\, \{ \la\in\C\, :\,  \dim(\ker(S-\la))=\infty \text{ or } \codim(\range(S-\la))=\infty \},\\
   \disspec(S)\, &:=\, \{ \la\in\C\, :\, \la \text{ is an isolated eigenvalue of S with finite multiplicity} \}.
\end{align*}
For a selfadjoint operator $S$ we have $\disspec(S) = \spectrum(S)\setminus\essspec(S)$.
\par\medskip

\begin{corollary}\label{cor:schurvar:essspec}
   In addition to the assumptions of Proposition~\thmref{prop:schurvar:schur} suppose that the operator $\mT$ is selfadjoint and that $T_{12}^*$ is surjective.
   Furthermore, assume that the operator function $\SCmin{1}$ defined in~\eqref{eq:SCmin} is a selfadjoint holomorphic operator family.
   Then we have 
   \begin{align}
      \label{eq:schurvar:pointspec}
      \pointspec(\mT)\cap \rho(T_{22})\ &=\ \pointspec(\SCmin{1}),\\
      \label{eq:schurvar:essspec}
      \essspec(\mT)\cap(\Dup,\infty)\ &=\ \essspec(\SCmin{1}) \cap(\Dup,\infty).
   \end{align}
\end{corollary}

\begin{proof}
   Proposition~\thmref{prop:schurvar:spectrum} applied to $\mT-\la$ shows that $\la\in\spectrum(\mT)\cap(\Dup,\infty)$ if and only if $\la\in\spectrum(\SCmin{1})$.
   Moreover, it follows from Proposition~\thmref{prop:schurvar:spectrum} that  $\la\in\pointspec(\SCmin{1})$ if and only if $\la\in\pointspec(\mT)\cap (\Dup,\infty)$ with  $\dim(\ker(\mT-\la)) = \infty$ if and only if $\dim(\ker(\SCmin{1}(\la)))=\infty$.
   Hence,~\eqref{eq:schurvar:pointspec} is proved.
   By assumption, for each $\la\in(\Dup,\infty)$, the operator $\SCmin{1}(\la)$ is selfadjoint.
   Therefore, in order to show~\eqref{eq:schurvar:essspec}, it suffices to show $\disspec(\mT)\cap(\Dup,\infty) = \disspec(\SCmin{1})\cap(\Dup,\infty)$.
   Let $\la\in \disspec(\mT)\cap(\Dup,\infty)$. Then we have $\dim\range(\SCmin{1}(\la))^\perp =  \dim\ker(\SCmin{1}(\la)) = \dim\ker (T -\la) < \infty$. 
   Further, the range of $\mT-\la$ is closed because $\la\in\disspec(\mT)$. So Proposition~\thmref{prop:schurvar:spectrum} shows that $\range(\SCmin{1}(\la)) = \range(\mT)\cap (\hil_1\oplus\{0\})$ is also closed. 
   Hence it follows $0\in\disspec(\SCmin{1}(\la))$ and consequently $\la\in\disspec(\SCmin{1})$.
   
   Let $\la\in\disspec(\SCmin{1})$.
   Then $\la\in\pointspec(\mT)$ with $\dim\ker(\mT-\la) = \dim\ker(\SC{1}(\la))<\infty$ and we have to show that $\la$ is no accumulation point of $\spectrum(\mT)$. 
   Since $0\in\disspec(\SCmin{1}(\la))$ and $\SCmin{1}$ is holomorphic, 
   there are $\delta>0$, $\eps>0$ and holomorphic functions $\mu_j:(\la-\delta,\la+\delta)\rightarrow\R$ with $\mu_j(\la)=0$ for $j=1,\,\dots,\, \dim\ker(\SC{1}(\la))$,
   such that  for all $\widetilde\la\in(\la-\delta,\, \la+\delta)$ we have that 
   $\mu\in\spectrum(\SCmin{1}(\widetilde\la))\cap (-\eps, \eps)$ if and only if $\mu$ is an eigenvalue of $\SCmin{1}(\widetilde\la)$ with finite multiplicity and $\mu=\mu_j(\widetilde\la)$ for some $j$
   (see~\cite[chap. IV, \S3 and chap. VII]{kato}).
   Furthermore, for $ j=1,\,\dots,\, \dim\ker(\SCmin{1}(\la))$ we have 
   \begin{align*}
      \diff{\la}\mu_j(\la)\ =\ \diff{\la} \scalar{x_j}{ \SCmin{1}(\la)x_j}\ 
      =\  -\|x_j\|^2 - \|(T_{22}-\la)^{-1}T_{12}^*\,x_j \|^2 
      \ <\ 0
   \end{align*}
   for normalised eigenvectors $x_j$ of $\SCmin{1}(\la)$ with eigenvalue $0$,
   hence the functions $\mu_j$ are not constant in a neighbourhood of $\la$.
   Consequently, there exists a nonempty interval $(\la-\widetilde\delta,\ \la+\widetilde\delta)$ such that 
   $0\in\rho(\SCmin{1}(\widetilde\la))$ for all $\widetilde\la\in(\la-\widetilde\delta,\ \la+\widetilde\delta)\setminus\{\la\}$.
   Consequently, $\spectrum(\mT)\cap (\la-\widetilde\delta,\ \la+\widetilde\delta) = \{\la\}$ which completes the proof.
\end{proof} 

\begin{remark}
   In this section we have considered only the case $\la\in\rho(T_{22})$.
   For $\la\in \rho(T_{11})$ the Schur complement
   \begin{align*}
      \mD(\SCmin{2}(\la))\, &:=\,
      \{ \psi\in\mD(T_{12})\, :\, (T_{11}-\la)^{-1}T_{12}\psi\in\mD(T_{12}^*) \},\\[1ex]
	 \SCmin{2}(\la)\, &:=\, 
	 T_{22}-\la - T_{12}^*(T_{11}-\la)^{-1}T_{12}
   \end{align*}
   acting on the Hilbert space $\hil_2$ can be used to obtain statements analogous to those given for the Schur complements $\SC{1}(\la)$.
\end{remark}


\section{Variational principle for the block operator matrix $\mT$} 
\label{sec:variatonal}
On the Hilbert space $\hil=\hil_1\oplus\hil_2$ we consider the block operator matrix $\mT=\tm{T_{11}}{T_{12}}{T_{12}^*}{T_{22}}$.
We assume that the conditions \ref{B1}, \ref{A1}, \ref{A2}, \ref{D1}, \ref{D2a} and \ref{mT1} hold.
Let $\SC{1}(\la)$, $\la\in\rho(T_{22})$, be the Schur complement of $\mT$ (cf. Proposition~\thmref{prop:schurvar:schur} and Proposition~\ref{prop:schurvar:inclusion}).

The next propositions summarise the properties of the Schur complements and its associated forms.

\begin{proposition}\label{prop:schurvar:summary}
   Consider the selfadjoint block operator matrix
   \begin{align*}
      \mT=\tm{T_{11}}{T_{12}}{T_{12}^*}{T_{22}},
      \hspace{3ex} 
      \mD(\mT)=\mD(T_{12}^*)\oplus\mD(T_{12})
   \end{align*}
   on the Hilbert space
   $\hil_1\oplus\hil_2$.
   Assume that the conditions \ref{B1}, \ref{A1}, \ref{A2}, \ref{D1} and  \ref{D2a} hold. 

   \begin{enumerate}
      \item\label{item:schurvar:formdomain}
      For every $\la\in(\Dup,\infty)$, the form
      \begin{align*}
	 \mD(\SCfs{1}(\la))\, &:=\, \mD(T_{12}^*),\\
	 \SCfs{1}(\la)[u,\, v]\, &:=\, \scalar{u}{(T_{11}-\la)v} - \bigl(\scalar{T_{12}^*u}{(T_{22}-\la)^{-1}T_{12}^*v}\bigr)
      \end{align*}
      is closed and its domain is independent of $\la$.
      The operator $\SC{1}(\la)$ associated with the form $\SCfs{1}(\la)$ is a well defined selfadjoint operator and $\SCmin{1}(\la)=\SC{1}(\la)$, $\la\in(\Dup,\infty)$.
   \end{enumerate}
   Define the operator valued function
   \begin{align}\label{eq:schurvar:S2}
      \SC{1}: (\Dup,\infty) \longrightarrow \closed(\hil_1),
      \quad
      \lambda \mapsto \SC{1}(\la)
   \end{align}
   and, for fixed $x\in\mD(\SCfs{1})$, the function 
   \label{first:sigmax}
   \begin{align}\label{eq:schurvar:s2}
      \sig{1}^x: (\Dup,\infty)\longrightarrow \R,\qquad
      \sig{1}^x(\la)= \SCfs{1}(\la)[x]
   \end{align}
   where $\closed(\hil_1)$ denotes the set of all closed operators on $\hil_1$.
   \begin{enumerate}
      \setcounter{enumi}{1}
      
      \item The operator valued function $\SC{1}:(\Dup,\infty)\rightarrow\closed(\hil_1)$ of~\eqref{eq:schurvar:S2} is continuous in the norm resolvent topology, and for every $x\in\mD(\SCfs{1})$ the function 
      $\sig{1}^x$ defined in~\eqref{eq:schurvar:s2} is continuous.
   
      \item \label{item:schurvar:decrease}
      For every $x\in\mD(\SCfs{1})\setminus\{0\}$ the function
      $\sig{1}^x$ is decreasing and unbounded from below. 
   
      \item If additionally the condition
      \begin{condlist}
	 \CondItem{{\rm(\BTwo)}}
	 $T_{12}^*$ is surjective 
	 \label{B2}
      \end{condlist}
      is satisfied, then it follows that
      $\essspec(\SC{1})=\essspec(\mT)\cap(\Dup,\infty)$  and
      $\pointspec(\SC{1})=\pointspec(\mT)\cap(\Dup,\infty)$.
   \end{enumerate}
\end{proposition}

\begin{proof}
   \begin{proofenumerate}
      \item The assertions concerning $\SCfs{1}(\la)$ are shown in Proposition~\thmref{prop:schurvar:form}
   while the identity $\SCmin{1}(\la)=\SC{1}(\la)$ was proved in Proposition~\thmref{prop:schurvar:inclusion}.
   In particular, the mapping $\SC{1}$ is well defined.
      
   \item From \ref{item:schurvar:formdomain} it follows that the family of sesquilinear forms $(\SCfs{1}(\la))_{\la\in(\Dup,\infty)}$ is of type~(a) according to the classification in~\cite{kato}.
      Hence $\SC{1}$ is a holomorphic family of type~(B), 
      which implies the holomorphy  of $\SC{1}$ in the norm resolvent topology. Obviously, for every $x\in\mD(\SCfs{1})$ the function $\sig{1}^x$ is even smooth on $(\Dup,\infty)$.
      
      \item For every $x\in\mD(\SCfs{1})$, $x\neq 0$,  the function $\sig{1}^x$ is monotonously decreasing since
      \begin{align}\label{eq:schurvar:deriv}
	 \diff{\la}\sig{1}^x(\la)\ =\
	 \diff{\la}\SCfs{1}(\la)[x]\ =\ 
	 -\|x\|^2 - \|(T_{22}-\la)^{-1}T_{12}^*x\|^2\ \le\ -\|x\|^2 \ <\ 0.
      \end{align}

      \item This has been shown in Corollary~\thmref{cor:schurvar:essspec}.
      \qedhere
   \end{proofenumerate}
\end{proof}

\begin{proposition}\label{prop:schurvar:summaryadd}
   Suppose that in addition to the assumptions of Proposition~\thmref{prop:schurvar:summary} there is a constant $b>0$ such that for all $x\in\mD(T_{12}^*)$ the estimate
   \begin{align}\label{eq:schurvar:b}
      \| T_{12}^*x\| \ge b \|x\|
   \end{align}
   holds. 
   For $\la\in(\Dup,\infty)$ let $d(\la)$ be a nonnegative lower bound for $(\la-T_{22})^{-1}$, i.e., 
   \begin{align}\label{eq:schurvar:dla}
      \bigscalar{x}{(\la-T_{22})^{-1}x}\ \ge\ d(\la)\|x\|^2\ge 0,\qquad
   x\in\hil_2,\ \la\in(\Dup,\infty).
   \end{align}
   If there is a $\delta>0$ with 
   \begin{align}\label{eq:schurvar:est1}
      \delta \ <\  d(\la)b^2 + \Alo - \la
   \end{align}
   for all $\la$ in a sufficiently small right neighbourhood $(\Dup,\, \Dup+\eps)$ of $\Dup$, then
   \begin{enumerate}
   \setcounter{enumi}{4}
   \item\label{item:schurvar:specspace}
   the spectral subspace $\specspace_{(-\infty,\,0)}\SC{1}(\la)$ is trivial
   for all $\la\in (\Dup,\, \Dup+\eps)$;
   \item  $\essspec(\SC{1})\cap(\Dup,\, \Dup+\eps)=\emptyset$.
   \end{enumerate}
   If we allow $\delta=0$ in equation~\eqref{eq:schurvar:est1}, then we can show \ref{item:schurvar:specspace} only.
\end{proposition}

\begin{proof}
   For $\la\in(\Dup,\, \Dup+\eps)$, assumptions~\eqref{eq:schurvar:dla} and~\eqref{eq:schurvar:est1} imply for all $x\in\mD(\SC{1}(\la))\setminus\{0\}$ that
   \begin{align}
      \nonumber
      \scalar{x}{ \SC{1}(\la)x}\ &=\ \SCfs{1}(\la)[x]\ =\
      \scalar{x}{T_{11}x} - \la\|x\|^2 + \bigscalar{T_{12}^*x}{(\la-T_{22})^{-1}T_{12}^*x}\\
      \label{eq:schurvar:s2pos}
      &\ge\ (\Alo-\la)\|x\|^2 +  d(\la)\,b^2\,\|x\|^2\
      >\ \delta\, \|x\|^2.
   \end{align}
   \begin{proofenumerate}
   \setcounter{proofenumi}{4}
   \item 
   If $\delta\ge 0$, then for all $\la\in(\Dup,\ \Dup+\eps)$
   the numerical range of the selfadjoint operator $\SC{1}(\la)$, the closure of which equals the closure of the numerical range of $\SCfs{1}(\la)$, is contained in the right half plane $\{z\in\C\, :\, \re(z) \ge 0\}$, implying $\rho(\SC{1}(\la))\supseteq(-\infty,\ 0)$. 
   \item 
   If we assume the strict inequality $\delta>0$, then the calculation above shows that $(-\infty, \delta)\subseteq\rho(\SC{1}(\la))$ for $\la\in(\Dup, \Dup+\eps)$, hence $(\Dup, \Dup+\eps)\cap \spectrum(\SC{1}) = \emptyset$.
   \qedhere
   \end{proofenumerate}
\end{proof}

Proposition~\thmref{prop:schurvar:summary}~\ref{item:schurvar:decrease} shows that for every $x\in\mD(\SCfs{1})\setminus\{0\}$ the function  $\sig{1}^x$ has at most one zero and is unbounded from below. 
If in addition \eqref{eq:schurvar:est1} holds with some $\delta>0$, then $\sig{1}^x$ is positive for $\la$ in a sufficiently small right neighbourhood of $\Dup$, see \eqref{eq:schurvar:s2pos}.
Thus the continuity of $\sig{1}^x$ implies that it has exactly one zero. We denote this zero by $p(x)$, i.e.,
\begin{align}\label{eq:schurvar:pdef}
   \sig{1}^x(\la) = 0\quad\Longleftrightarrow\quad
   \la = p(x).
\end{align}
If relation \eqref{eq:schurvar:est1} does not hold, then the function $\sig{1}^x$ does not need to have a zero. 
In this case we define $p(x):=-\infty$, so that obviously either $p(x)=-\infty$ or 
$p(x)>\Dup$. 
Further, $p(x)$ does not depend on the norm of $x$, i.e., for all $\xi\in\C\setminus \{0\}$ we have $p(x)=p(\xi x)$.

\par\medskip   

Now fix a linear manifold $\mD\subseteq\hil_1$, independent of $\la$,  such that 
\begin{align*}
   \mD(\SC{1}(\la))\ \subseteq\ \mD\ \subseteq\ \mD(\SCfs{1}(\la)),
   \qquad\qquad\la\in(\Dup,\infty).
\end{align*}
Such a manifold $\mD$ exists; for example, we can choose $\mD=\mD(T_{12}^*)$.

For $n\in\N$ we define the numbers 
\label{first:muEigenvalues}
\begin{align}\label{first:noz}
   \mu_n\ :=\ \min_{\msub{L\subseteq\mathcal D}{\dim L = n}}\ \max_{x\in \noz{L}}\ p(x),
\end{align}
where 
$\noz{L}:=L\setminus\{0\}$. 
Theorem~\thmref{theorem:schurvar:var1} shows that these numbers are indeed well defined.
Here and in the following, a sequence $\la_1\le\la_2\le\dots\le \la_N$ with $N=\infty$ has to be understood as the infinite sequence $\la_1\le\la_2\le\dots$\ .

For an interval $\Delta\subseteq\R$ and a selfadjoint operator $S$ we denote its spectral subspace corresponding to $\Delta$ by $\specspace_{\Delta}(S)$. 
By $\lae$ we denote the lower bound of the essential spectrum of $\SC{1}$, i.e.,
\begin{align*}
   \lae\ :=\ \begin{cases} 
      \inf\essspec(\SC{1})\quad\quad &\text{if $\essspec(\SC{1})\neq\emptyset$}, \\
      \infty\qquad &\text{if $\essspec(\SC{1})=\emptyset$}.
   \end{cases}
\end{align*}
\par\medskip\par
If $(\Dup,\, \lae)$ is not empty, then the eigenvalues of $\SC{1}$ in this interval are characterised by the following minimax principle.

\begin{theorem}\label{theorem:schurvar:var1}
   Let the block operator matrix $\mT=\tm{T_{11}}{T_{12}}{T_{12}^*}{T_{22}}$ with domain $\mD(\mT)=\mD(T_{12}^*)\oplus\mD(T_{12})$ be self\-adjoint in the Hilbert space $\hil_1\oplus\hil_2$.
   Suppose that the conditions \ref{B1}, \ref{A1}, \ref{A2}, \ref{D1} and \ref{D2a} are satisfied and that $T_{12}^*$ is surjective.
   Further, assume that the set $(\Dup,\,\lae)$ is nonempty and that there is a $\la_0\in(\Dup,\, \lae)$ such that $\dim\specspace_{(-\infty,0)}\SC{1}(\la_0)<\infty$.
   \label{first:n0}
   Then the index shift
   \begin{align}
      n_0:=\min\limits_{\la>\Dup}\dim\specspace_{(-\infty,0)}\SC{1}(\la)
   \end{align}
   is finite and $\spectrum(\mT)\cap(\Dup, \lae)$ consists of a {\upshape(}possibly infinite{\upshape)} sequence of eigenvalues $\la_1\le\la_2\le\dots\le \la_N$ where $N\in\N_0\cup\{\infty\}$.
   If the eigenvalues are counted according to their multiplicity, then
   \begin{align}\label{eq:schurvar:lafirst}
      \la_{\eignumber} = \mu_{\eignumber+n_0}, \qquad 1\le \eignumber\le N,
   \end{align}
   and $N\in\N_0\cup\{\infty\}$ is given by
   \begin{align*}
      N\ =\ \dimlae(\lae) - n_0
   \end{align*}
   where $\dimlae(\lae)$ is the dimension of maximal subspaces of the set
   \begin{align*}
      \{x\in\mD\ :\ \exists\ \la>\Dup \text{ with } \SCfs{1}(\la)[x]<0 \}
      \cup\{0\}.
   \end{align*}
   If $N=\infty$, then $\lim\limits_{n\rightarrow\infty}\la_\eignumber =\lae$. 
   If $N<\infty$ and $\essspec(\SC{1})=\emptyset$,  then $\mu_\eignumber=\infty$ for $\eignumber>n_0+N$.
   If $N<\infty$, $\lae<\infty$,  then $\mu_\eignumber=\lae$ for $\eignumber>n_0+N$.

   If there exists a $\delta$ as in Proposition~\thmref{prop:schurvar:summaryadd}, then $n_0=0$.
\end{theorem}
\begin{proof}   
   Proposition~\thmref{prop:schurvar:summary} 
   shows that all assumptions of Theorem~\cite[theorem 2.1]{eschwe} are satisfied for the Schur complement $\SC{1}(\la)$, $\la\in(\Dup,\infty)$.
Hence, the numbers $\mu_{\eignumber+n_0}$ exist and are equal to the eigenvalues of the operator family $\SC{1}$. 
   By Corollary~\thmref{cor:schurvar:essspec}, we have
   $\pointspec(\SC{1})=\pointspec(\mT)\cap(\Dup,\infty)$ and $\essspec(\SC{1})=\essspec(\mT)\cap(\Dup,\infty)$ so that all the assertions follow from theorem~\cite[theorem 2.1]{eschwe}.

   If even the assumptions of Proposition~\thmref{prop:schurvar:summaryadd} are valid, then it follows automatically that $(\Dup,\lae)\neq \emptyset$ and that $\dim \specspace_{(-\infty, 0)}\SC{1}(\la)=0$ for $\la$ in a sufficiently small right neighbourhood of $\Dup$,
   hence the index offset $n_0$ appearing in formula~\eqref{eq:schurvar:lafirst} vanishes. 
\end{proof} 

The numbers $p(x)$ are rather hard to estimate. 
However, there is a representation of $p(x)$ as the supremum of a functional $\la_+\tv{x}{y}$ where $y$ varies in some subspace of $\hil_2$, see~\eqref{eq:schurvar:la} and Lemma~\ref{lemma:schurvar:p}.
The functional $\la_+$ is connected with the so-called quadratic numerical range of block operator matrices, see, for example, \cite{LT98} and \cite{LMMT01}.
It was used in~\cite{LLT02} to obtain a variational principle for block operator matrices with bounded off-diagonal entries.

\begin{definition}
\label{first:QNR}
\label{defininition:schurvar:laFunctional}
   \index{quadratic numerical range}
   \index{QNR}
   Let $\mT=\tm{T_{11}}{T_{12}}{T_{21}}{T_{22}}$ be a closed block operator matrix on the Hilbert space $\hil=\hil_1\oplus\hil_2$ with domain $\mD(\mT)=\mD(T_{21})\oplus\mD(T_{12})$. 
      Assume that the operators $T_{12}$ and $T_{21}$ are closed and that $T_{11}$ is $T_{21}$-bounded and that $T_{22}$ is $T_{12}$-bounded. 
      For $\ttv{x}{y}\in\mD(\mT)$, $x\neq0$, $y\neq 0$, consider the matrices
      \begin{align*}
	 \mT_{x,y}\ :=\ 
	 \begin{pmatrix}
	    \frac{\scalar{x}{T_{11}x}}{\|x\|^2} & \frac{\scalar{x}{T_{12}y}}{\|x\|\,\|y\|}\\[2ex]
	    \frac{\scalar{y}{T_{21}x}}{\|x\|\,\|y\|} & \frac{\scalar{y}{T_{22}y}}{\|y\|^2}
	 \end{pmatrix}\ \in\ M_{2}(\C)
      \end{align*}
      with eigenvalues
      \begin{align}\label{eq:schurvar:la} 
	 \nonumber
	 \la_{\pm}\mv{x\\y}\ :=\ 
	 \frac{1}{2}\Biggl( &
	       \frac{\scalar{x}{T_{11}x}}{\|x\|^2} + \frac{\scalar{y}{T_{22}y}}{\|y\|^2} 
	      \\
	      & \hspace{2ex}\pm \sqrt{\ts
	       \left(\frac{\scalar{x}{T_{11}x}}{\|x\|^2}-\frac{\scalar{y}{T_{22}y}}{\|y\|^2}\right)^2
	       + \frac{4\scalar{x}{T_{12}y}\scalar{y}{T_{21}x}}{\|x\|^2\,\|y\|^2} } \
	   \Biggr)
      \end{align}
      and define the sets
      \begin{align*}
	 \Lambda_{\pm}(\mT)\ :=\ \left\{ \la_{\pm}\mv{x\\y}\ :\
	    x\in\noz{\mD(T_{21})},\ y\in\noz{\mD(T_{12})} \right\}
      \end{align*}
      where we use again the notation $\noz L := L \setminus\{0\}$ for linear spaces $L$.
      Note that in general the radicand may be negative or complex.
      The \define{quadratic numerical range} $W^2(\mT)$ of $\mT$ is defined as the set of all complex numbers $\la$ that are eigenvalues of some $\mT_{x,y}$, that is,
      \begin{align*}
	 W^2(\mT)\ :=\ 
	 \bigcup\limits_{\msub{x\in\noz{\mathcal D(T_{21})} }{y\in\noz{\mathcal D(T_{12})} }}
	 \pointspec(\mT_{x,y})\
	 =\ \Lambda_+(\mT) \cup \Lambda_-(\mT). 
      \end{align*}
\end{definition}

The number $\la_{\pm}\tv{x}{y}$ does not depend on the norm of the vectors $x$ and $y$. 
Therefore it suffices to restrict the definition of $\la_\pm\tv{x}{y}$ to elements $\ttv{x}{y}\in\mD(\mT)$ with $\|x\|=\|y\|=1$.

In the following we characterise $p(x)$, defined in~\eqref{eq:schurvar:pdef}, in terms of $\la_\pm\tv{x}{y}$. 
Recall that $p(x)$ is the unique zero of the function $\la\mapsto \sig{1}^x(\la)$ if it exists and $p(x)= - \infty$ otherwise.

\begin{lemma}\label{lemma:schurvar:p}
   Assume that the conditions of Proposition~\thmref{prop:schurvar:summary} 
   hold.
   Then for all  $x\in\mD(T_{12}^*)\setminus\{0\}$ with $p(x)\neq -\infty$ we have
   \begin{align}\label{eq:schurvar:pSup}
      p(x)\ =\ \sup\Bigl\{\la_+\begin{pmatrix}x\\y\end{pmatrix}\ :\ y\in\mD(T_{12})\setminus\{0\}\Bigr\}.
   \end{align}
   If in addition $x\in\mD(\SC{1}(p(x)))$, then the supremum is attained, thus we have
   \begin{align}\label{eq:schurvar:pMax}
      p(x)\ =\ \max\Bigl\{\la_+\begin{pmatrix}x\\y\end{pmatrix}\ :\ y\in\mD(T_{12})\setminus\{0\}\Bigr\}.
   \end{align}
\end{lemma}
\begin{proof}
   Fix $x\in\mD(T_{12}^*)\setminus\{0\}$.
   Since $T_{21}=T_{12}^*$ and the operators $T_{11}$ and $T_{22}$ are symmetric, ~\eqref{eq:schurvar:la} shows that $\la_+\tv{x}{y}$ is real for all $y\in\mD(T_{12})\setminus\{0\}$.
   Note that for arbitrary $\la>\Dup$, $x\in\mD(T_{12}^*)\setminus\{0\}$, $y\in\mD(T_{12})\setminus\{0\}$ we have
   \begin{align}\label{eq:detSigma}
      \begin{aligned}
	 &\|x\|^2\,\|y\|^2\,\det(\mT_{x,y}-\la)\\
	 &\hspace{5ex}=\, \bigl( \scalar{x}{T_{11}x}-\la\|x\|^2 \bigr)
	 \bigl(\scalar{y}{T_{22}y}-\la\|y\|^2\bigr)
	 -\scalar{x}{T_{12}y}\scalar{y}{T_{12}^*x} 
      \end{aligned}
   \end{align}
   To prove the assertion we first show that $p(x)\ge \la_+\tv{x}{y}$ for all $y\in\mD(T_{12})\setminus\{0\}$. 
   So fix $y\in\mD(T_{12})\setminus\{0\}$ and, for simplicity of notation, set $\la_+ := \la_+\tv{x}{y}$. 
   If $\la_+\le \Dup$, then nothing has to be shown since $\Dup \le p(x)$ by assumption. 
   Now assume $\la_+>\Dup$.
   Since $p(x)$ is the unique zero of the monotonously decreasing function $\sig{1}^x$, it suffices to show $\sig{1}^x(\la_+)=\SCfs{1}(\la_+)[x]\ge 0$. 
   By definition, $\la_+$ is an eigenvalue of the complex $2\times2$-matrix $\mT_{x,y}$, thus, by~\eqref{eq:detSigma} and the definition of $\SCfs{1}(\la)$ in Proposition~\ref{prop:schurvar:form}:
   \begin{align}\nonumber
      0\ &=\ \|x\|^2\,\|y\|^2\,\det(\mT_{x,y}-\la_+)\\[1ex]
      \label{eq:schurvar:det1}
      &=\ \bigscalar{y}{(T_{22}-\la_+)y}\,\SCfs{1}(\la_+)[x] \\
      \label{eq:schurvar:det2}
      & \phantom{=\ (}+
      \bigscalar{y}{(T_{22}-\la_+)y}
      \bigscalar{T_{12}^*x}{(T_{22}-\la_+)^{-1}T_{12}^*x} - |(y,\,T_{12}^*x)|^2.
   \end{align}
   For $\la>\Dup$ the operator $(\la-T_{22})$ is strictly positive and the same holds for the induced sesquilinear form 
   $(u,v)\mapsto (u,\, (\la-T_{22})v)$ for $u,\, v\in\mD(T_{22})$. 
   For this form we have the following generalised Cauchy-Schwarz inequality 
   \begin{align*}
      |\, \scalar{u}{(\la-T_{22})v}\, |^2\
      &=\ \bigl|\, \bigscalar{(\la-T_{22})^{\frac{1}{2}}u}{(\la-T_{22})^{\frac{1}{2}}v}\, \bigr|^2\\
      &\le\ \bigl\| (\la-T_{22})^{\frac{1}{2}}u \bigr\|^2\, 
      \bigl\|(\la-T_{22})^{\frac{1}{2}}v \bigr\|^2\\
      &=\ \bigscalar{u}{(\la-T_{22}) u}  
      \bigscalar{v}{(\la-T_{22})v}
   \end{align*}
   for all $u,\,v\in\mD(T_{22})$.
   Since $y\in\mD(T_{12})\subseteq\mD(T_{22})$, we can use this inequality to estimate the two terms in~\eqref{eq:schurvar:det2}:
   \begin{align*}
      &\bigscalar{y}{(T_{22}-\la_+)y}
      \bigscalar{T_{12}^*x}{(T_{22}-\la_+)^{-1}T_{12}^*x} - |\scalar{y}{T_{12}^*x}|^2\\
      &\hspace{15ex}
      \begin{aligned}
	 &=\ 
	 \bigscalar{y}{(T_{22}-\la_+)y}
	 \bigscalar{T_{12}^*x}{(T_{22}-\la_+)^{-1}T_{12}^*x} \\
	 &\phantom{=\ (}- \bigl|
	 \bigscalar{y}{(\la_+-T_{22})(\la_+-T_{22})^{-1}T_{12}^*x}
	 \bigr|^2\\
	 &\ge\ 
	 \bigscalar{y}{(\la_+-T_{22})y}
	 \bigscalar{T_{12}^*x}{(\la_+-T_{22})^{-1}T_{12}^*x} \\
	 &\phantom{=\ (}
	 - \bigscalar{y}{(\la_+-T_{22})y}
	 \bigscalar{T_{12}^*x}{(\la_+-T_{22})^{-1}T_{12}^*x}\\
	 &=\ 0.
      \end{aligned}
   \end{align*}
   Because the factor $\scalar{y}{(T_{22}-\la_+)y}$ in the term~\eqref{eq:schurvar:det1} is negative, it follows that the second factor, $\SCfs{1}(\la_+)[x]=\sig{1}^x(\la_+)$, must be nonnegative, and thus we have proved the inequality $p(x) \ge \sup\{\la_+\tv{x}{y}\, :\, y\in\mD(T_{12})\setminus\{0\}\}$.
   
   If $x\in\mD(\SC{1}(p(x)))$, then we can choose an element $y$ such that $p(x)=\la_+\tv{x}{y}$.
   To this end, define $y:= (T_{22}-p(x))^{-1}T_{12}^*x$. 
   This vector is well defined and it lies in the domain of $T_{12}$ since by assumption $x\in\mD(\SC{1}(p(x)))$.
   If we use
   \begin{align}
      \label{eq:schurvar:zero}
      \bigscalar{y}{(T_{22}-p(x))y}
      \bigscalar{T_{12}^*x}{(T_{22}-p(x))^{-1}T_{12}^*x} 
      - \bigl| \bigscalar{T_{12}^*x}{\ y} \bigr|^2\ =\ 0
   \end{align}
   and $\SCfs{1}(p(x))[x]= 0$, 
   it follows with the help of~\eqref{eq:detSigma}
   \begin{align*}
      \|x\|^2\,\|y\|^2\, & \det(\mT_{x,y}-p(x))\\[1ex]
      &=\ \bigscalar{y}{(T_{22}-p(x))y}\,\SCfs{1}(p(x))[x] \\
      & \phantom{=\ (}+ 
      \bigscalar{y}{(T_{22}-p(x))y}
      \bigscalar{T_{12}^*x}{(T_{22}-p(x))^{-1}T_{12}^*x} - |(y,\,T_{12}^*x)|^2\\[1ex]
      &=\ 0.
   \end{align*}
   This implies that $p(x)$ is an eigenvalue of $\mT_{x, y}$.
   Together with $\la_-\tv{x}{y}\le\la_+\tv{x}{y}\le p(x)$ it follows that $p(x)=\la_+\tv{x}{y}$
   which proves~\eqref{eq:schurvar:pSup} and \eqref{eq:schurvar:pMax} in the case $x\in\mD(\SC{1}(p(x)))$.\\
   It remains to show~\eqref{eq:schurvar:pSup} in the case $x\notin\mD(\SC{1}(p(x)))$, i.e., for elements $x\in\mD(T_{12}^*)$ such that $(T_{22}-p(x))^{-1}T_{12}^*x\notin\mD(T_{12})$.
   If $x\in\mD(T_{12}^*)\setminus\mD(\SC{1}(p(x)))$, there exists a sequence
   $(x_n)_{n\in\N}\subseteq\mD(\SC{1}(p(x)))$ such that
   \begin{align*}
      x_n\rightarrow x 
      \quad\text{and}\quad 
      (T_{22}-\la)^{-\frac{1}{2}}T_{12}^*x_n\rightarrow (T_{22}-\la)^{-\frac{1}{2}}T_{12}^*x,
      \qquad n\rightarrow\infty,
   \end{align*}
   since in the proof of Proposition~\thmref{prop:schurvar:schur} we saw that $\mD(\SC{1}(p(x)))$ is a core of $(T_{22}-p(x))^{-\frac{1}{2}}T_{12}^*$.
   Set $y_n:=(T_{22}-p(x))^{-1}T_{12}^*x_n$, $n\in\N$. 
   Because both $(T_{22}-p(x))^{-1}$ and $T_{22}-p(x)$ are bounded, the limites $y:=\lim\limits_{n\rightarrow \infty} y_n$ and $\lim\limits_{n\rightarrow\infty}T_{12}^* x_n$ exist and are not zero; otherwise it would follow that $x\in\mD(\SC{1}(p(x)))$ in contradiction to the assumption on $x$.
   Moreover, since $T_{11}$ is relatively bounded with respect to $T_{12}^*$, also the limit $\lim\limits_{n\rightarrow\infty} T_{11}x_n$ exists.
   Therefore, all terms in 
   \begin{multline*}
      \SCfs{1}(p(x))[x_n]\\[1ex]
      \begin{aligned}[t]
      =\ & \SCfs{1}(p(x))[x] + 
      \SCfs{1}(p(x))[x_n-x, x_n] + \SCfs{1}(p(x))[x_n, x_n-x] - \SCfs{1}(p(x))[x_n-x]\\[1ex]
      =\ &
	 2\, \bigscalar{x_n-x}{(T_{11}-p(x))x_n} 
	 -  \bigscalar{x_n-x}{(T_{11}-p(x))(x_n-x)}\\
	 &+ 2\, \bigscalar{T_{12}^*(x_n-x)}{(T_{22}-p(x))^{-1}T_{12}^*x_n} \\
	 & - \bigscalar{T_{12}^*(x_n-x)}{(T_{22}-p(x))^{-1}T_{12}^*(x_n-x)}
      \end{aligned}
   \end{multline*}
   converge to zero for $n\rightarrow \infty$.
   As in~\eqref{eq:schurvar:zero}, we obtain
   \begin{align*} 
      \bigscalar{y_n}{(T_{22}-p(x))y_n}
      \bigscalar{T_{12}^*x_n}{(T_{22}-p(x))^{-1}T_{12}^*x_n} 
      - \bigl| \bigscalar{T_{12}^* x_n}{y_n} \bigr|^2\, 
      =\, 0,
      \quad n\in\N,
   \end{align*}
   which implies
   \begin{align*}
      \|x_n\|^2\, \|y_n\|^2\, \det(\mT_{x_n, y_n} - p(x))
      = \bigscalar{y_n}{(T_{22}-p(x))y_n}\, \SCfs{1}(p(x))[x_n]\
      \xrightarrow{\ n\to\infty\ }\ 0.
   \end{align*}
   Since neither $x_n$ nor $y_n$ tend to zero, it follows that
   \begin{align}\label{eq:schurvar:pproduct}
      \Bigl( p(x) - \la_-\tv{\vphantom{f_n}x_n}{\vphantom{f_n}y_n} \Bigr) 
      \Bigl( p(x) - \la_+\tv{\vphantom{f_n}x_n}{\vphantom{f_n}y_n} \Bigr) \ =\ 
      \det (\mT_{x_n, y_n} - p(x) )\
      \xrightarrow{\ n\to\infty\ }\ 0.
   \end{align}
   Each entry of 
   \begin{align*}
	 &\hspace{-2ex}\mT_{x_n, y_n} - \mT_{x, y}\\
      &=\ 
      \begin{pmatrix}
	 \scalar{x_n}{T_{11}(x_n-x)} + \scalar{x_n-x}{T_{11}x} &
	 \scalar{x_n}{T_{12}(y_n-y)} + \scalar{x_n-x}{T_{12}y} \\
	 \scalar{y_n}{T_{12}^*(x_n-x)} + \scalar{y_n-y}{T_{12}^*x} &
	 \scalar{y_n}{T_{22}(y_n-y)} + \scalar{y_n-y}{T_{22}y}
      \end{pmatrix}\\[2ex]
      &=\ 
      \begin{pmatrix}
	 \scalar{x_n}{T_{11}(x_n-x)} + \scalar{x_n-x}{T_{11}x} &
	 \scalar{T_{12}^*x_n}{y_n-y} + \scalar{T_{12}^*(x_n-x)}{y} \\
	 \scalar{y_n}{T_{12}^*(x_n-x)} + \scalar{y_n-y}{T_{12}^*x} &
	 \scalar{y_n}{T_{22}(y_n-y)} + \scalar{y_n-y}{T_{22}y}
      \end{pmatrix}
   \end{align*}
   converges to zero for $n\rightarrow\infty$, hence we have $\mT_{x_n,y_n}\rightarrow\mT_{x,y}$ in norm.
   Thus the eigenvalues $\la_\pm\tv{\vphantom{f_n}x_n}{\vphantom{f_n}y_n}$ of $\mT_{x_n, y_n}$ converge to the eigenvalues $\la_\pm\tv{\vphantom{f_n}x}{\vphantom{f_n}y}$ of $\mT_{x, y}$, in particular it follows that 
   \begin{align*}
      p(x) = 
      \lim\limits_{n\to\infty}
      \la_+\tv{\vphantom{f_n}x_n}{\vphantom{f_n}y_n}.
   \end{align*}
   Since $\lambda_\pm$ is continuous in both its independent variables, and since $y_n\in\mD(T_{12}^*$, it follows that 
   \begin{align}\label{eq:schurvar:plimes}
      p(x) = 
      \lim\limits_{n\to\infty}
      \la_+\tv{\vphantom{f_n}x}{\vphantom{f_n}y_n} 
      \le 
      \sup\Bigl\{ \la_+\tv{\vphantom{f}x}{\vphantom{f}y}\, :\, y\in\mD(T_{12}^*)\setminus\{0\} \Bigr\}.
   \end{align}

\end{proof}

The following theorem is the main theorem of this paper.
It provides a variational characterisation of the eigenvalues of $\mT$ in a right half plane of $\C$ in terms of its entries $T_{ij}$.

\begin{theorem}\label{theorem:schurvar:var2}
   Suppose that the assumptions of Theorem~\thmref{theorem:schurvar:var1} hold,
   that is, suppose that conditions \ref{mT1}, \ref{B1}, \ref{A1}, \ref{A2}, \ref{D1} and \ref{D2a} are fulfilled, that $T_{12}^*$ is surjective, that $(\Dup,\, \lae)\neq\emptyset$ and that there is a $\la_0\in (\Dup,\lae)$ such that $\dim\specspace_{(-\infty, 0)}\SC{1}(\la_0)<\infty$.
   Then the eigenvalues of $\mT$ in $(\Dup,\lae)$ are given by
   \begin{align}\label{eq:schurvar:minmaxsup}
	 \la_{\eignumber}\ 
	 &=\  
	 \minp_{ \msub{L\subseteq \mathcal D(T_{12}^*)}{\dim L = \eignumber+n_0} }\ 
	 \maxp_{x\in \noz{L}}\ 
	 \sup_{y\in\noz{\mathcal D(T_{12})}} \la_+\mv{x\\y},
	 \qquad\qquad 1\le \eignumber\le N,
   \end{align}
   where we have adopted the notation of Theorem~\thmref{theorem:schurvar:var1}.
   If the domain of $\SC{1}(\la)$ does not depend on $\la$, i.e., if
   \begin{align*}
      \mD(\SC{1}(\la))\ =\ \mD(\SC{1}), \qquad\qquad \la\in(\Dup,\infty)
   \end{align*}
   then we have
   \begin{align}\label{eq:schurvar:minmaxmax}
	 \la_{\eignumber}\ 
	 =\  
	 \min_{ \msub{L\subseteq \mathcal D(\SC{1})}{\dim L = \eignumber+n_0} }\ 
	 \max_{x\in \noz{L}}\ 
	 \max_{y\in\noz{\mathcal D(T_{12})}} \la_+\mv{x\\y},
	 \qquad\qquad 1\le \eignumber\le N.
   \end{align}
\end{theorem}

\begin{proof}
   From Theorem~\thmref{theorem:schurvar:var1} it follows that all eigenvalues of $\mT$ greater than $\Dup$ are given by
   \begin{align*}
      \la_{\eignumber}\ =\  
      \min\limits_{ \msub{L\subseteq \mathcal D}{\dim L=\eignumber+n_0} }\ 
      \max\limits_{x\in \noz{L}}\ p(x),
      \qquad\qquad 1\le \eignumber\le N,
   \end{align*}
   where $\mD$ is any linear manifold with $\mD(\SC{1}(\la))\subseteq\mD\subseteq\mD(\SCfs{1})$.
   From Proposition~\thmref{prop:schurvar:summary} we know that the forms $\SCfs{1}(\la)$, $\la\in(\Dup,\infty)$, are closed and that $\mD(\SCfs{1}(\la))=\mD(T_{12}^*)$.
   Fix $\eignumber>0$ and a subspace $L\subseteq\mD(\SCfs{1})$ with $\dim L = \eignumber+n_0$.
   Then there exists an $x\in L$ with $p(x)\neq -\infty$.
   Lemma~\thmref{lemma:schurvar:p} yields
   \begin{align*}
      \max_{x\in \noz{L}}\ p(x)\ &=\ 
      \max_{\msub{x\in \noz{L}}{p(x)\neq-\infty}}\ p(x)\\
      & =\ 
      \max_{\msub{x\in \noz{L}}{p(x)\neq-\infty}}\ \sup_{y\in\noz{\mathcal D(T_{12})}} \la_+\mv{x\\y}\ =\
      \max_{x\in \noz{L}}\ \sup_{y\in\noz{\mathcal D(T_{12})}} \la_+\mv{x\\y}.
   \end{align*}
   If we have even $L\subseteq\mD(\SC{1})$, then the supremum can be replaced by the maximum.
\end{proof}

In general, it is not easy to determine the index shift $n_0$.
Sufficient conditions for the finiteness $n_0$, which are met the operator in the application in Section~\ref{sec:application}, are given in the following proposition.

\begin{proposition}\label{prop:schurvar:n0finite}
   \begin{enumerate}

      \item 
      If in addition to the assumptions in Theorem~\thmref{theorem:schurvar:var2} the operator $T_{12}^*$ is compactly invertible and the operator $T_{11}$ is bounded, then the index shift $n_0$ is finite.

      \item 
      If the assumptions of Proposition~\ref{prop:schurvar:summaryadd} are satisfied, then $n_0=0$.

   \end{enumerate}

\end{proposition}

\begin{proof}
   Let $\la\in(\Dup,\infty)$.
   \begin{proofenumerate}

      \item 
      The assumptions imply that $T_{12}(T_{22}-\la)^{-1}T_{12}^*$ is compactly invertible.
      Since $T_{11}$ is bounded, also $\SC{1}(\la)$ is compactly invertible, hence its spectrum consists of a sequence of eigenvalues with finite multiplicity which has no accumulation point.
      Since the operator $\SC{1}(\la)$ is bounded from below, it follows that $\dim\specspace_{(-\infty,0)} \SC{1}(\la) < \infty$, in particular,
      \begin{align*}
	 n_0 = \min_{\la>\Dup}\dim\specspace_{(-\infty,0)} \SC{1}(\la) 
	 < \infty.
      \end{align*}

      \item
      Proposition~\ref{prop:schurvar:summaryadd} implies that 
      $\specspace_{(-\infty,0)} \SC{1}(\la) = \emptyset$ for $\la$ sufficiently close to $\Dup$, hence
      $n_0 = \min_{\la>\Dup}\dim_{\specspace_{(-\infty,0)}} \SC{1}(\la) = 0$.
      \qedhere
   \end{proofenumerate}
\end{proof}

In the application in section~\ref{sec:application}, the spectrum of $T_{12}T_{12}^*$ consists of simple discrete eigenvalues only. 
For this situation, we specialise Theorem~\thmref{theorem:schurvar:var2} further.

\begin{remark}\label{remark:schurvar:squares}
   Assume that $\mT_0=\tm{0}{T_{12}}{T_{12}^*}{0}$ with domain $\mD(\mT_0) = \mD(T_{12}^*)\oplus\mD(T_{12})\subseteq\hil_1\oplus\hil_2$ is closed and that $T_{12}T_{12}^*$ and $T_{12}^*T_{12}$ are strictly positive. 
   Then $\pointspec(\mT_0)\ =\ \{ \la\in\R\, :\, \la^2\in\pointspec(T_{12}T_{12}^*) \}$.
\end{remark}
\begin{proof}
      For $\la\in\pointspec(\mT_0)\setminus\{0\}$ we have $\la^2\in\pointspec(T_{12}T_{12}^*)\cap\pointspec(T_{12}^*T_{12})$ since for 
      each eigenvector $\ttv{f}{g}$ of $\mT_0$ with eigenvalue $\la$ it follows that $f\in\mD(T_{12}T_{12}^*)$, $g\in\mD(T_{12}^*T_{12})$, $f,\, g\neq 0$ and 
   $0=(\mT_0+\la)(\mT_0-\la)\tv{f}{g} 
   = \bigl( \tm{T_{12}T_{12}^*}{0}{0}{T_{12}^*T_{12}} -\la^2 \bigr)\tv{\vphantom{T_1}f}{\vphantom{T_1}g}
   = \tv{(T_{12}T_{12}^*-\la^2)f}{(T_{12}^*T_{12}-\la^2)g}$.\\
   On the other hand, if $\mu\neq 0$ is an eigenvalue of $T_{12}^*T_{12}$ with eigenfunction $g$, then it is also an eigenvalue of $T_{12}T_{12}^*$ with eigenfunction $T_{12}g$. For $\sigma=\pm 1$  we define $f=\sigma\mu^{-\frac{1}{2}}\, T_{12}g$. Then we have that $(\mT_0-\sigma\sqrt{\mu})\tv{f}{g}=0$, hence $\pm\sqrt\mu$ are eigenvalues of $\mT_0$.
\end{proof}

To estimate the functionals  $\la_+\tv{x}{y}$, we use the following auxiliary lemma.

\begin{lemma}\label{lemma:schurvar:analysisII}
   For $a_1,\ a_2,\ b_1,\ b_2,\ \ga\in\R$ with $a_1<b_1$ and $a_2<b_2$ we define the function
   \begin{align*}
      f:[a_1,b_1]\times[a_2,b_2]\rightarrow\R,\ 
      f(s,t)= s+t+\sqrt{(s-t)^2+\ga^2}.
   \end{align*}
   For fixed $t$, the function $f$ is monotonously increasing in $s$ and vice versa. 
   In particular,
   \begin{align*}
      f(a_1,a_2)\ \le\ f(s,t)\ \le\ f(b_1,b_2), \qquad (s,t)\in[a_1,b_1]\times[a_2,b_2].
   \end{align*}
\end{lemma}
\begin{proof}
   Partial differentiation of $f$ with respect to $s$ yields
   \begin{align*}
      \pa{s}f(s,t)\ =\ 
      1+\frac{s-t}{\sqrt{ (s-t)^2+\ga^2  }}\
      \ge\ \frac{\sqrt{(s-t)^2+\ga^2}-|s-t|}{\sqrt{ (s-t)^2+\ga^2}}\
      \ge\ 0.
      &\qedhere
   \end{align*}
\end{proof}

\begin{theorem}\label{theorem:schurvar:estimate}
   Let $\mT=\tm{T_{11}}{T_{12}}{T_{12}^*}{T_{22}}$ with domain $\mD(T_{12})\oplus\mD(T_{12}^*)\subseteq\hil_1\oplus\hil_2$ 
   be a selfadjoint block operator matrix such that the conditions \ref{mT1}, \ref{B1}, \ref{A1}, \ref{A2}, \ref{D1} and \ref{D2a} hold.
  Then $T_{11}$ is bounded with respect to $T_{12}^*$; let $\ax$, $\abx$, $\Dclo$ and $\Dcup$ such that 
   \begin{align*}
      \|T_{11}x\|\ &\le\ \ax \|x\| + \abx \|T_{12}^*x\|, 
      && x\in\mD(T_{12}^*),
      \\[1ex]
      \Dclo\,\|y\|^2\ &\le\ \scalar{y}{T_{22}y}\ \le\ \Dcup\,\|y\|^2,
      && y\in\hil_2.
   \end{align*}
   Further, let $T_{12}^*$ be bijective
   and assume that for all $\la\in(\Dup,\infty)$ the Schur complement
   \begin{align*}
      \mD(\SC{1}(\la))\ &=\
      \{x\in\mD(T_{12}^*)\, :\, (T_{22}-\la)^{-1}T_{12}^*x\in\mD(T_{12})\},\\
      \SC{1}(\la)\ &=\ T_{11}-\la -T_{12}(T_{22}-\la)^{-1}T_{12}^*,
   \end{align*}
   is selfadjoint and that $\mD(\SC{1}(\la))=:\mD(\SC{1})$ is independent of $\la$.
   Additionally suppose that there exists a $\la_0\in(\Dup,\infty)$ such that $\dim\specspace_{(-\infty,0)}\SC{1}(\la_0)<\infty$.
   If the spectrum of the operator $T_{12}T_{12}^*$ satisfies
   \label{first:BB*Eigenvalues}%
   \begin{align*}
      \spectrum(T_{12}T_{12}^*)\ =\
      \pointspec(T_{12}T_{12}^*)\ =\ \{\nu_j\, :\, j\in\N\}\
      \quad\text{with}\quad
      0<\nu_1\le\nu_2\le\dots\
   \end{align*}
   where the eigenvalues are counted with their multiplicities, 
   then the block operator matrix $\mT$ has discrete point spectrum $\la_1\le\la_2\le\dots\la_N$  in $(\Dup, \lae)$. 
   More precisely, if $n_0$ is as in Theorem~\thmref{theorem:schurvar:var1}, that is,
   $n_0 = \min\limits_{\la>\Dup}\dim\specspace_{(-\infty,0)}\SC{1}(\la)$,
   then for all $ 1\le \eignumber\le N$
   the eigenvalues $\la_\eignumber$ of\ $\mT$ in $(\Dup, \lae)$ satisfy the estimates 
   \begin{alignat}{2}
      \label{eq:schurvar:lale}
      \la_{\eignumber}\ &\le\
      \frac{\abx}{2} \sqrt{\nu_{\eignumber+n_0}}  + 
      \sqrt{ \nu_{\eignumber+n_0} + {\ts\frac{1}{4}} (\abx\sqrt{\nu_{\eignumber+n_0}}+ |\ax-\Dcup|)^2  }
      \ + \ts\frac{1}{2} \left( \ax + \Dup  \right)
      ,\quad
      \\[2ex]
      \label{eq:schurvar:lage}
      \la_{\eignumber}\ &\ge\
      \sqrt{\nu_{\eignumber+n_0} + {\ts\frac{1}{4}} (\Aclo - \Dclo)^2 }\
      + \ts\frac{1}{2}(\Alo + \Dclo).
   \end{alignat}
%

\end{theorem}

\begin{proof}
   Since $T_{12}^*$ is closed, its resolvent is bounded by the closed graph theorem.
   Hence Proposition~\ref{prop:schurvar:summaryadd} yields that $(\Dup, \lae)\neq\emptyset$ so that all assumptions of Theorem~\thmref{theorem:schurvar:var2} are satisfied.
   In particular, the index shift $n_0$ is finite.
   To prove inequalities~\eqref{eq:schurvar:lale} and \eqref{eq:schurvar:lage},
   we estimate the right hand side of \eqref{eq:schurvar:minmaxmax}. 
   Note that $\mD(T_{22}) = \hil_2$ and that  
   \begin{enumerate}
      \item\label{item:schurvar:T11}
      $\scalar{x}{T_{11}x}\ \le\ |\scalar{x}{T_{11}x}|\ \le\ \|x\|\, \|T_{11}x\|\ \le\ \|x\|\,(\ax\|x\|+\abx\|T_{12}^*x\|)$, \quad $x\in\mD(T_{12}^*)$,

      \item\label{item:schurvar:Aclo}
      $\scalar{x}{T_{11}x}\ \ge\ \Aclo \|x\|^2$, \quad $x\in\mD(T_{11})$,

      \item\label{item:schurvar:Dup}
      $\scalar{y}{ T_{22}y}\ \le\ \Dup\|y\|^2, \quad y\in\hil_2$,

      \item\label{item:schurvar:Dclo}
      $\scalar{y}{ T_{22}y}\ \ge\ \Dclo\|y\|^2, \quad y\in\hil_2$,

      \item\label{item:schurvar:CS}
      $ |\scalar{y}{T_{12}^*x}|^2\ \le \|y\|^2\, \|T_{12}^*x\|^2$, \quad $x\in\mD(T_{12}^*)$,\ \ $y\in\hil_2$.
	 
   \end{enumerate}

   First we prove~\eqref{eq:schurvar:lale}.
   With the help of inequalities
   \ref{item:schurvar:T11}, \ref{item:schurvar:Dup} and \ref{item:schurvar:CS}
   and the auxiliary Lemma~\ref{lemma:schurvar:analysisII} we find for 
   $x\in\mD(T_{12}^*)\setminus\{0\}$, $y\in\mD(T_{12})\setminus\{0\}$:
   \begin{align*}
      \la_+\!\mv{x\\y}
      & = \frac{1}{2}\left(
	 \frac{\scalar{x}{T_{11}x}}{\|x\|^2} + \frac{\scalar{y}{T_{22}y}}{\|y\|^2} + 
	 \sqrt{ \biggl(\frac{\scalar{x}{T_{11}x}}{\|x\|^2} - \frac{\scalar{y}{T_{22}y}}{\|y\|^2}\biggr)^2 + \frac{4|\scalar{y}{T_{12}^*x}|^2}{\|x\|^2\, \|y\|^2} }\
      \right)
      \\[1ex]
      &\le\ \frac{1}{2} \left(
      \ax + \frac{\abx\|T_{12}^* x\|}{\|x\|} + \Dup + 
      \sqrt{ \biggl(\ax+ \frac{\abx\|T_{12}^*x\|}{\|x\|} - \Dup \biggr)^2
      + \frac{4|\scalar{y}{T_{12}^*x}|^2}{\|x\|^2\, \|y\|^2}
      }\
      \right)
      \\[1ex]
      &\le\ \frac{1}{2} \left( \ax + \Dup + \frac{\abx\|T_{12}^* x\|}{\|x\|} + 
      \sqrt{ \biggl( \frac{\abx\|T_{12}^*x\|}{\|x\|} + |\ax - \Dup| \biggr)^2 + \frac{4\,\|T_{12}^*x\|^2}{\|x\|} }\
      \right).
   \end{align*}
   The right hand side is independent of $y$ and monotonously increasing in $\|T_{12}^*x\|$.
   For given $n\in\N$ let $\mathscr L_n$ be an $n$-dimensional subspace of the spectral space $\specspace_{[\nu_1,\nu_n]}(T_{12}T_{12}^*) $. 
   Then for every $x\in\mathscr L_n$ we have that $\|T_{12}^*x\|^2=\scalar{x}{T_{12}T_{12}^*x}\le \nu_n\|x\|^2$. 

   Observe that $\mD(\SCfs{1})=\mD(T_{12})^*$, thus the minimax principle \eqref{eq:schurvar:minmaxmax} shows that 
   \begin{align*}
      \la_\eignumber\ &= 
      \minp_{ \msub{L\subseteq\mathcal D(T_{12})^*}{\dim L = \eignumber+n_0} }\ 
      \maxp_{x\in \noz{L}}\ 
      \sup_{y\in\noz{\mathcal D(T_{12})}}
      \la_+\mv{x\\y}\
      \le 
      \maxp_{ x\in \noz{\mathscr L_{\eignumber+n_0}} }\ 
      \sup_{y\in\noz{\mathcal D(T_{12})} } \la_+\mv{x\\y}\\[2ex]
      &\le\ \max_{x\in \noz{\mathscr L_{\eignumber+n_0}} }
      \frac{1}{2} \Biggl( \ax + \Dup + \frac{\abx \|T_{12}^*x\|}{\|x\|}
      + \sqrt{ \biggl(\frac{\abx\|T_{12}^*x\|}{\|x\|} + |\ax- \Dup| \biggr)^2
	 + \frac{4\,\|T_{12}^*x\|^2}{\|x\|^2}  }\hspace{1ex}
      \Biggr)\\[2ex]
      &\le\ 
      \frac{1}{2} \left( \ax + \Dup + \abx\sqrt{\nu_{\eignumber+n_0}} + 
	 \sqrt{  (\abx\sqrt{\nu_{\eignumber+n_0}}+ |\ax-\Dup| )^2 + 4\,\nu_{\eignumber+n_0}
	 }\,
      \right)
   \end{align*}
   which proves~\eqref{eq:schurvar:lale}.
   Note that if $\abx\sqrt{\nu_{n_0}} \ge \ax - \Dup$, then the above estimates is true with $\ax -\Dup$ instead of $|\ax -\Dup|$.

   In order to show~\eqref{eq:schurvar:lage}, we choose a particular $y\in\mD(T_{12})$. 
   Since by assumption $T_{12}^{*-1}$ exists and is bounded by $b^{-1}$,  also $T_{12}^{-1}$ exists and is bounded by $b^{-1}$.
   For every $x\in\mD(T_{12}^*)$ the element $y(x):=T_{12}^{-1}x$ exists and lies in $\mD(T_{12})$. 
   Therefore, again by~\eqref{eq:schurvar:minmaxmax} and the inequalities
   \ref{item:schurvar:Aclo} and \ref{item:schurvar:Dclo} we obtain
   \begin{align}
      \nonumber
      \la_\eignumber\ &= 
      \minp_{ \msub{L\subseteq\mathcal D(T_{12}^*)}{\dim L = \eignumber+n_0} }\ 
      \maxp_{x\in \noz{L}}\ 
      \sup_{y\in\noz{\mathcal D(T_{12})}}
      \la_+\mv{x\\y} \\[1ex]
      \label{eq:schurvar:T22}
      &\ge 
      \minp_{ \msub{L\subseteq\mathcal D(T_{12}^*)}{\dim L = \eignumber+n_0} }\ 
      \maxp_{x\in \noz{L}}\ 
      \sup_{y\in\noz{\mathcal D(T_{12})}}\
      \frac{1}{2}\bigg(
      \Alo + \Dclo 
      \sqrt{ (\Alo - \Dclo)^2 +  \frac{4\,|\scalar{x}{T_{12}y}|^2}{\|x\|^2\, \|y\|^2} }
      \bigg)\\[1ex]
      \nonumber
      &\ge 
      \minp_{ \msub{L\subseteq\mathcal D(T_{12}^*)}{\dim L = \eignumber+n_0} }\ 
      \max_{x\in \noz{L}}\ 
      \frac{1}{2}\bigg(
      \Alo + \Dclo 
      + \sqrt{ (\Alo - \Dclo)^2 + \frac{4\,\scalar{x}{x}^2}{\|T_{12}^{-1}x\|^2\,\|x\|^2} }
      \bigg)\\[1ex]
      \nonumber
      &=
      \minp_{ \msub{L\subseteq\mathcal D(T_{12}^*)}{\dim L = \eignumber+n_0} }\ 
      \max_{x\in \noz{L}}\ 
      \frac{1}{2}\bigg(
      \Alo + \Dclo 
      + \sqrt{ (\Alo - \Dclo)^2 + 4\, \|T_{12}^{-1}x\|^{-2} \|x\|^2 }
      \bigg)
      \\[1ex]
      \label{eq:schurvar:laBelow}
      &=\ 
      \frac{1}{2}( \Alo + \Dclo )
      + \sqrt{ {\ts\frac{1}{4}}(\Alo - \Dclo)^2 +
      \Big(\smash{
      \minp_{ \msub{L\subseteq\mathcal D(T_{12}^*)}{\dim L = \eignumber+n_0} }\  }
      \max_{x\in \noz{L}}\ 
      \|T_{12}^{-1}x\|^{-1} \|x\| \Big)^2 }
      .
   \end{align}
For every $n$-dimensional subspace $L_n\subseteq\mD(T_{12}^*)$, also $T_{12}^{-1}L_n\subseteq\mD(T_{12}^*T_{12})$ is $n$-dimen\-sional. 
   Hence it follows that 
   \begin{align*}
      \minp_{ \msub{L\subseteq\mathcal D(T_{12}^*)}{\dim L = \eignumber+n_0} }\ 
      \max_{x\in\noz{L}}\ 
      \|T_{12}^{-1}x\|^{-1}\,\|x\|
      \ &=\ 
      \minp_{ \msub{L\subseteq\mathcal D(T_{12}^*)}{\dim L = \eignumber+n_0} }\ 
      \max_{\xi\in T_{12}^{-1}\noz{L}}\ 
      \|\xi\|^{-1}\,\|T_{12}\xi\| \\
      \ &=\ 
      \min_{\msub{ L\subseteq\mathcal D(T_{12}^*T_{12}) }{\dim L = \eignumber+n_0} }\ 
      \max_{\xi\in \noz{L}}\ 
      \|\xi\|^{-1}\,\|T_{12}\xi\|.
   \end{align*}
   The squares of the nonzero eigenvalues of 
   $\widetilde\mT_0=\tm{0}{T_{12}^*}{T_{12}}{0} = 
   \tm{0}{\vphantom{T_{11}}\id}{\vphantom{T_{11}}\id}{0} \mT_0 
   \tm{0}{\vphantom{T_{11}}\id}{\vphantom{T_{11}}\id}{0} $
   are the eigenvalues $\nu_1\le\nu_2\le\dots$ of $T_{12}T_{12}^*$
   (see Remark~\thmref{remark:schurvar:squares}).
   On the other hand, the variational principle of Theorem~\thmref{theorem:schurvar:var2} applied to $\widetilde\mT_0$ shows that
   \begin{align*}
      \sqrt{\nu_n}\ =\ \la_n\ &=\ 
      \min_{ \msub{L\subseteq\mathcal D(T_{12}^*T_{12})}{\dim L=n} }\
      \max_{\xi\in \noz{L}}\ \max_{y\in\noz{\mathcal D(T_{12}^*)}}\
      \frac{|\scalar{y}{T_{12}\xi}|}{\|y\|\,\|\xi\|}  \\
      &\le\ 
      \min_{ \msub{L\in\mathcal D(T_{12}^*T_{12})}{\dim L=n} }\
      \max_{\xi\in \noz{L}}\ 
      \frac{|\scalar{T_{12}\xi}{T_{12}\xi}|}{\|T_{12}\xi\|\,\|\xi\|}  
      \ =\ 
      \min_{ \msub{L\in\mathcal D(T_{12}^*T_{12})}{\dim L=n} }\
      \max_{\xi\in \noz{L}}\ 
      \frac{\|T_{12}\xi\|}{\|\xi\|}.
   \end{align*}
   Inserting into \eqref{eq:schurvar:laBelow} yields
   \begin{align*}
      \la_\eignumber\ \ge\ \ts\frac{1}{2}(\Alo + \Dclo)
      + \sqrt{
      {\ts\frac{1}{4}}(\Alo - \Dclo)^2 + \nu_{\eignumber+n_0}}.
      &\qedhere
   \end{align*}

\end{proof}

\begin{remark}
Theorem~\thmref{theorem:schurvar:estimate} can be regarded as a perturbation result for the eigenvalues of the block operator matrix $\tm{0}{T_{12}}{T_{12}^*}{0}$ under the unbounded perturbation $\tm{T_{11}}{0}{0}{T_{22}}$ since in the case $T_{11}=T_{22}=0$ the spectral shift $n_0$ vanishes  and the estimates~\eqref{eq:schurvar:lale} and \eqref{eq:schurvar:lage} reduce to $\la_\eignumber = \sqrt{\nu_\eignumber}$.
If the sequence $(\nu_n)_n$ of the eigenvalues of $T_{12}T_{12}^*$ is unbounded, then $\la$ has the same asymptotics as $\sqrt{\nu_n}$,
i.e., $\frac{\la_n}{\sqrt{\nu_{n+n_0}} } \to 1$ for $n\to\infty$.
\end{remark}

If also the operator $T_{11}$ is bounded, then the estimate for $\la_n$ from above can be further improved.

\begin{theorem}\label{theorem:schurvar:var3}
   In addition to the assumptions in Theorem~\thmref{theorem:schurvar:estimate}, let $T_{11}$ and $T_{22}$ be bounded.
   Then there are real numbers $\Aclo$ and $\Acup$ such that 
   \begin{alignat*}{3}
      \Aclo\,\|x\|^2\ &\le\ \scalar{x}{T_{11}x}\ &\le&\ \Acup\,\|x\|^2,
      \qquad 
      &x\in\hil_1.
   \end{alignat*}
   Let $n_0=\min\limits_{\la>\Dcup}\dim\specspace_{(-\infty,0)}\SC{1}(\la)$.
   Then the eigenvalues of the block operator matrix $\mT$ in $(\Dcup,\, \lae)$, 
   enumerated such that $\Dcup < \la_1\le \la_2\le \dots $, 
   can be estimated by
   \begin{alignat}{2}
      \label{eq:schurvar:Est}
      \la_\eignumber\ &\le\
      \sqrt{\nu_{\eignumber+n_0} + \ts\frac{1}{4} (\Acup - \Dcup)^2 }\
      + \ts\frac{1}{2} ( \Acup + \Dcup ),
      \qquad &1\le \eignumber\le N, \\[2ex]
      \la_\eignumber\ &\ge\
      \sqrt{\nu_{\eignumber+n_0} + \ts\frac{1}{4} (\Aclo - \Dclo)^2 }\
      + \ts\frac{1}{2} (\Aclo + \Dclo) ,
      \qquad &1\le \eignumber\le N.
   \end{alignat}
   where $0\, <\, \nu_1\, \le\, \nu_2\, \le\, \dots $ are the eigenvalues of $T_{12}^{}T_{12}^*$, see Theorem~\thmref{theorem:schurvar:estimate}.
   The index shift $n_0$ is given by $n_0=\min\limits_{\la>\Dcup}\dim\specspace_{(-\infty,0)} \SC{1}(\la)$.
\end{theorem}

\begin{proof}
   We only need to show \eqref{eq:schurvar:Est}.
   If we use $\scalar{x}{T_{11}x} \le \Acup\,\|x\|^2$, $x\in\hil_1$
   instead of inequality \ref{item:schurvar:T11} 
   in the proof of formula~\eqref{eq:schurvar:lale},
   we obtain with the help of the auxiliary Lemma~\thmref{lemma:schurvar:analysisII}
    \begin{align*}
      \la_+\mv{x\\y}\!
      &= \frac{1}{2}\Biggl( \!\scalar{x}{T_{11}x} + \scalar{y}{T_{22}y}
	    +
	 \sqrt{ ( \scalar{x}{T_{11}x}-\scalar{y}{T_{22}y} )^2 + 4|\scalar{y}{T_{12}^*x}|^2}
      \Biggr)\\[1ex]
      &\le\ \frac{1}{2}(\Acup +\Dcup) + 
      \sqrt{ \ts\frac{1}{4} (\Acup - \Dcup)^2  +\|T_{12}^*x\|^2}
   \end{align*}
   for all $\ttv{x}{y}\in\mD(\mT)$ with $\|x\|=\|y\|=1$.
   Now formula \eqref{eq:schurvar:Est} follows by a reasoning analogous to that of the proof of \eqref{eq:schurvar:lale}.
\end{proof}

\section{Application to the angular part of the Dirac equation in the Kerr-Newman background metric} 
\label{sec:application}

The Kerr-Newman metric describes the spacetime in the exterior of an electrically charged rotating massive black hole.
A spin-$\frac{1}{2}$ particle with mass $m$ and electrical charge $e$ outside the black hole obeys the Dirac equation
\begin{align}\label{eq:DE}
   (\widehat\FA + \widehat\FR)\widehat\Psi\, =\, 0
\end{align}
where $\widehat\Psi$ is a four-component wave function describing the particle and $\widehat\FA$ and $\widehat\FR$ are $4\times 4$ differential expressions that contain partial derivatives with respect to all four spacetime coordinates.
It can be shown that by a suitable ansatz the Dirac equation~\eqref{eq:DE} can be decoupled into a system of two ordinary differential equations (\cite{chandrasekhar}, \cite{thesis}, \cite{WY06}): the radial equation that contains only derivatives with respect to the radial coordinate and the angular equation that contains only derivatives with respect to the angular coordinate $\theta$.
For recent results on the radial equation see~\cite{Schmid} and \cite{WY06}.
The angular equation is given by
\begin{align*}
   (\FA-\la)\Psi = 0 
   \hspace{4ex}\text{on } 
   \hspace{1ex} (0,\pi).
\end{align*}
with the differential expression 
\begin{align}\label{eq:FA}
   \FA\, =\, 
   \begin{pmatrix}
      -am\cos\theta & \FB_+ \\
      \FB_- & am\cos\theta 
   \end{pmatrix}
   \quad\text{ on }\quad (0,\pi),
\end{align}
where
\begin{align}
   \FB_\pm\, =\, \pm\diff{\theta} + \frac{k+\frac{1}{2}}{\sin\theta} + a\om\sin\theta.
\end{align}
The number $k\in\Z$ describes the motion of the electron in the plane of symmetry. The parameter $a := J/M \in\R$ describes the rotation of the black hole where $J$ is the angular momentum and $M$ is the mass of the black hole.

For the following results on the angular operator we refer to~\cite{thesis}.
Let 
\begin{align*}
   \hil\, :=\, \Ltwo{(0,\pi)}{\rd\theta} \times \Ltwo{(0,\pi)}{\rd\theta}
\end{align*}
with the norm $\|\ttv{f}{g}\| = \sqrt{ \ltwo{f} + \ltwo{g} }\ $ where $\ltwo{\cdot}$ denotes the usual norm on $\Ltwo{(0,\pi)}{\rd\theta}$.
It can be shown that on $\hil$ the formal expression $\FA$ has the unique selfadjoint realisation
\begin{align}\label{eq:angschur:angularOp}
   \mA\ =\ \begin{pmatrix} -D & B \\B^* & D \end{pmatrix},\qquad
   \mD(\mA) = \mD(B^*)\oplus\mD(B)
\end{align}
where $D$ is the operator of multiplication  by the function $(0,\pi)\rightarrow\R,\ \theta\mapsto am\cos\theta$ and $B$ is the first order differential operator given by
\begin{align*}
   \mD(B) = \bigl\{f\in\hil\, :\, f \text{ is absolutely continuous, }\FB_+ f \in\hil \bigr\},
   \quad B f = \FB_+ f.
\end{align*}
$B$ is closed and its inverse can be computed explicitely.
It turns out that $B^{-1}$ is a Hilbert-Schmidt operator, hence it is compact.
Consequently, also the angular operator $\mA$ in the special case $a=0$ is compactly invertible.
Since all terms in $\mA$ involving the parameter $a$ are bounded, a perturbation argument shows that $\mA$ is compactly invertible for every value of $a$.
Furthermore, it can be shown that the spectrum of $\mA$ consists of simple eigenvalues only. 

\begin{remark}\label{remark:symmetry}
   It can be shown that if $\theta \mapsto (f(\theta), g(\theta))^t$ is an eigenfunction of $\mA$, then it follows that there is a $\gamma\in\C$ with $|\gamma| = 1$ such that $(g(\pi-\theta), f(\pi-\theta))^t = \gamma (f(\theta), g(\theta))^t$, $\theta\in(0,\pi)$.

   \par
   In the special case $a=0$ the eigenvalues can be calculated explicitly; one obtains
   \begin{align*}\ts
      \pointspec(\mA)\, =\, \{\sign(n)( |k+\frac{1}{2}| -\frac{1}{2} + n )\, :\, n\in\Z\setminus\{0\} \}
      \qquad\text{for}\quad a=0.
   \end{align*}
\end{remark}

The operator $BB^*$ is a Sturm-Liouville operator with spectrum consisting of discrete simple eigenvalues $0< \nu_1< \nu_2 < \dots $ only. 
By Sturm's comparison theorem we obtain the following two-sided estimates for the eigenvalues $\nu_n$ of $T_{12}T_{12}^*$:
\begin{align}\label{eq:sturm}
   \ts
   \max\{0,\, (|k+\frac{1}{2}|-\frac{1}{2}+n)^2 + \Omega_-\}
   \,\le\, \nu_n\le\,
   (|k+\frac{1}{2}|-\frac{1}{2}+n)^2 + \Omega_+
\end{align}
with 
\begin{align}\label{eq:Omega}
   \Omega_-\, &=\,\ts 2(k+\frac{1}{2}) a\om - | a\om|,& 
   \Omega_+\, &=\, \begin{cases}
      a^2\om^2 + \frac{1}{4} + 2(k+\frac{1}{2})a\om\hspace{3ex}
      &\text{if } 2a\om\notin [-1,\,1\,],\\[1ex]
      2(k+\frac{1}{2}) a\om + | a\om|
      &\text{if } 2a\om\in [-1,\,1\,].
   \end{cases}
\end{align}

\begin{remark}
   In the case $a=0$ these estimates give the correct eigenvalues of $BB^*$.
   The corresponding eigenfunctions are hypergeometric functions.
\end{remark}

For the block operator matrix $\mA$, the Schur complement for $\la\in\rho(D)$ is given by
\begin{align}\label{eq:angschur:schur}
   \mD(\SC{1}(\la))\ &=\ \{f\in\mD(B^*)\, :\, (D-\la)^{-1}B^* f\in\mD(B) \},\\
   \SC{1}(\la)\ &=\ -D-\la - B(D-\la)^{-1}B^*.
\end{align}

\begin{lemma}\label{lemma:angschur:prerequisites}
   The angular operator fulfils conditions \ref{mT1}, \ref{B1}, \ref{B2}, \ref{A1}, \ref{A2}, \ref{D1} and \ref{D2a} of the preceding section, in particular, we have 
   \begin{alignat}{6}
	 \label{eq:angschur:angopA1}\tag{\AOne$'$}
	 \Aclo\ &:=\ & -|am|\ &\le\
	 \scalar{x}{-Dx}\ &\le &\ |am|
	 &\ =:\ &\Acup,\qquad
	 &x&\in\hil\\[0.5ex]
	 \label{eq:angschur:angopD1}\tag{\DOne$'$}
	 \Dclo\ &:=\ & -|am|\ &\le\
	 \scalar{x}{Dx}\ &\le &\ |am|
	 &\ =:\ &\Dcup,
	 &x&\in\hil,\\[0.5ex]
	 \label{eq:angschur:angopA2}\tag{\ATwo$'$), (\DTwoa$'$}
	 &&&&\makebox[0cm][r]{$\|{-D}\,\|\ =\ \|D\,\|\ $}= &\ |am|\makebox[0cm][l]{,}
   \end{alignat}
   and $\spectrum(D)=\spectrum(-D)=\essspec(D) = [-|am|, |am|\,]$.
   For all $\la\in\mathbb (|am|,\infty)$, the form
   \begin{align}
      \mD(\SCfs{1}(\la)) = \mD(B^*),
      \qquad 
      \SCfs{1}(\la)[f,g]\ :=\ \scalar{f}{(-D-\la)g} - \scalar{B^*f}{(D-\la)^{-1}B^*g},
   \end{align}
   is symmetric, semibounded from below and closed. 
   Further, the operator $\SC{1}(\la)$ is the selfadjoint operator associated with $\SCfs{1}(\la)$, 
   and its domain is independent of $\la$, more precisely, we have
   \begin{align}\label{eq:angschur:SDdomain}
      \mD(\SC{1}(\la))\ =\ \mD(BB^*), 
      \qquad \la\in(|am|,\infty).
   \end{align}
\end{lemma}
\begin{proof}
   Since $D$ is the bounded operator given by multiplication with the continuous, nowhere constant function $am\cos\theta$, $\theta\in(0,\pi)$, the assertions concerning the spectrum of $D$ and relations~\eqref{eq:angschur:angopA1}, \eqref{eq:angschur:angopD1} and \eqref{eq:angschur:angopA2} are clear.
   Hence conditions \ref{A1} and \ref{D1} are satisfied with $\Dcup = |am|$ and $\Alo = -|am|$, 
   and \ref{A2} and  \ref{D2a} hold because $D$ is bounded with $\|D\|=|am|$.
   Since $\sigma(D)=[-|am|, |am|\,]$, the sesquilinear forms $\SCfs{1}(\la)$, $\la\in(|am|,\infty)$, are well defined, and, by Proposition~\thmref{prop:schurvar:form}, they are symmetric, semibounded from below and closed.
   Proposition~\thmref{prop:schurvar:inclusion} implies that for $\la\in(|am|,\infty)$ the operator  $\SC{1}(\la)$ is the selfadjoint operator associated with $\SCfs{1}(\la)$.\\
   To prove~\eqref{eq:angschur:SDdomain}, fix $f\in\mD(BB^*)$ and $\la\in(|am|, \infty)$. 
   We have to show $(D-\la)^{-1}B^*f\in\mD(B)$.
   Since both $(am\cos\theta -\la)^{-1}$  and $B^*f$ are absolutely continuous, we have
   \begin{multline*}
      \FB_+ (D-\la)^{-1}B^*f(\theta)\\
      \begin{aligned}
	 &=\
	 \Bigl( \diff{\theta} + \frac{k+\frac{1}{2}}{\sin\theta} + a\om\sin\theta \Bigr)
	 (am\cos\theta -\la)^{-1}
	 \Bigl( -\diff{\theta} + \frac{k+\frac{1}{2}}{\sin\theta} + a\om\sin\theta \Bigr) f(\theta) \\[1ex]
	 &=\ 
	 (am\cos\theta -\la)^{-1}
	 \Bigl( \diff{\theta} + \frac{k+\frac{1}{2}}{\sin\theta} + a\om\sin\theta \Bigr)
	 \Bigl( -\diff{\theta} + \frac{k+\frac{1}{2}}{\sin\theta} + a\om\sin\theta \Bigr) f(\theta)\\
	 &\qquad + \Bigl( \diff{\theta} (am\cos\theta -\la)^{-1} \Bigr) 
	 \Bigl( -\diff{\theta} + \frac{k+\frac{1}{2}}{\sin\theta} + a\om\sin\theta \Bigr) f(\theta)\\[1ex]
	 &=\ 
	 (D-\la)^{-1} BB^* f(\theta) +
	 \Bigl( \diff{\theta} (am\cos\theta -\la)^{-1} \Bigr) B^*f(\theta).
      \end{aligned}
   \end{multline*}
   Observe that the first term on the first line is the formal differential expression associated with $B$.
   Since, by assumption, $f\in\mD(BB^*)$ and since both $(D-\la)^{-1}$ and $\diff{\theta}(am\cos\theta-\la)^{-1}$ are bounded operators on $\hil$, it follows that $(D-\la)^{-1}B^*f\in\mD(B)$, and consequently $f\in\mD(\SC{1}(\la))$.\\
   Conversely, assume $f\in\mD(\SC{1}(\la))$ for some $\la\in(|am|, \infty)$.
   The function $am\cos\theta-\la$ is differentiable on $(0,\pi)$, hence we have
   \begin{multline*}
      \FB_+ B^*f\ =-\
      \Bigl( \diff{\theta} + \frac{k+\frac{1}{2}}{\sin\theta} + a\om\sin\theta \Bigr)
      B^*f(\theta)\\
      \begin{aligned}
	 &=\ 
	 \Bigl( \diff{\theta} + \frac{k+\frac{1}{2}}{\sin\theta} + a\om\sin\theta \Bigr)
	 (am\cos\theta-\la)(am\cos\theta -\la)^{-1} \\
	 & \phantom{=\ \bigl(} \times 
	 \Bigl( -\diff{\theta} + \frac{k+\frac{1}{2}}{\sin\theta} + a\om\sin\theta \Bigr) f(\theta)\\[1ex]
	 &=\ 
	 (am\cos\theta -\la)
	 \Bigl( \diff{\theta} + \frac{k+\frac{1}{2}}{\sin\theta} + a\om\sin\theta \Bigr)
	 (am\cos\theta -\la)^{-1}\\
	 &\hspace{11pt} \times 
	 \Bigl( -\diff{\theta} + \frac{k+\frac{1}{2}}{\sin\theta} + a\om\sin\theta \Bigr) f(\theta)\\
	 &\hspace{11pt} + 
	 \Bigl( \diff{\theta} (am\cos\theta -\la)\Bigr) 
	 (am\cos\theta -\la)^{-1} 
	 \Bigl( -\diff{\theta} + \frac{k+\frac{1}{2}}{\sin\theta} + a\om\sin\theta \Bigr) f(\theta)\\[1ex]
	 &=\ 
	 (D-\la)
	 B(D-\la)^{-1} B^* f(\theta) +
	 \Bigl( \diff{\theta} (am\cos\theta -\la) \Bigr)(D-\la)^{-1} B^*f(\theta).
      \end{aligned}
   \end{multline*}
   Since $D-\la$ and $\diff{\theta}(am\cos\theta-\la)$ are bounded operators on $\hil$, it follows that the function above is also an element of $\hil$, hence we have $B^*f\in\mD(B)$ which implies that $f\in\mD(BB^*)$.
\end{proof}


\subsection{Explicit bounds for the eigenvalues of $\mA$} 
\label{subsec:explicite}

As mentioned already earlier the spectrum of the angular operator consists only of isolated simple eigenvalues without accumulation points in $(-\infty,\infty)$.
We also know that the eigenvalues depend continuously on the parameter $a$. 
To express this dependence explicitly we frequently write $\la_n(a)$.
Therefore we can enumerate the eigenvalues $\la_{\Aeignumber}(a),\ \Aeignumber\in\Z\setminus\{0\}$, unambiguously by requiring that 
$\la_{\Aeignumber}(a)$ is the analytic continuation of $\la_n(0)=\sign(n)\bigl(|k+\frac{1}{2}|-\frac{1}{2}+ |n| \bigr)$ in the case $a=0$.
Since all eigenvalues are simple, it follows that $\la_n(a)<\la_m(a)$ for $n<m$.

For fixed Kerr parameter $a$ we define $m_\pm\in\Z$ such that
\label{first:mplus}
\begin{multline*}
   \dots\ \le\ \la_{m_--2}(a)\ \le\ \la_{m_--1}(a)\ <\ 
   -|am|\ \le\ \la_{m_-}(a)\ \le\ 
   \dots \\
   \dots
   \le\ \la_{m_+}(a)\ \le\ |am| 
   <\ \la_{m_++1}(a)\ \le\ \la_{m_++2}(a)\ \le\ \dots
\end{multline*}
i.e., $\spectrum(\mA)\cap [-|am|,|am|\,] = \{\la_{\Aeignumber}(a)\,:\, m_-\le \Aeignumber \le m_+,\ n\neq 0\}$ and the number of eigenvalues of $\mA$ in the interval $[-|am|, |am|\,]$ is given by
\begin{align*}
   \#\, \bigl(\sigma(\mA)\cap [-|am|, |am|\,]\,\bigr)\ =\
   \begin{cases}
      m_+-m_-\qquad  &\text{if } 0\in[m_-, m_+],\\
      m_+-m_-+1\quad\quad  &\text{if } 0\notin[m_-, m_+].
   \end{cases}
\end{align*}
Observe that  $m_+$ and $m_-$ depend on the physical parameters $a$, $m$, $\om$ and $k$.

\begin{theorem}\label{theorem:angschur:var1}
   Let $n_0=\dim\specspace_{(-\infty, 0)}\SC{1}(\la_0)$ for some $\la_0\in(|am|, \la_{m_++1})$.
   Then the eigenvalues of the angular operator $\mA$ to the right of $|am|$ are given by
   \begin{align}\label{eq:angschur:variaton}
      \la_{m_++\Aeignumber}(a)\ =\ 
      \min_{ \msub{L\subseteq\mathcal D(BB^*)}{\dim L = \Aeignumber+n_0} }\ 
      \max_{x\in \noz{L}}\ 
      \max_{y\in\noz{\mathcal D(B)}}\ \la_+\mv{x\\y},
      \qquad \Aeignumber\in\N. 
   \end{align}
   Furthermore, the eigenvalues can be estimated by
   \begin{align}
      \label{eq:angschur:Est}
      \sqrt{\nu_{n_0+\Aeignumber}} - |am|\ \le\ 
      \la_{m_++\Aeignumber}(a)\ &\le\ \sqrt{\nu_{n_0+\Aeignumber}} + |am| ,
      \qquad \Aeignumber\in\N, 
   \end{align}
   where $\nu_{\Aeignumber+n_0}$ are the eigenvalues of $BB^*$.
   Explicit estimates for $\la_n$ in terms of the physical parameters $a,\ m$ and $\om$ are 
   \begin{align*}
      \lalo_{n+n_0}\ \le\  \la_{m_++\Aeignumber}(a)\ \le\
      \laup_{n+n_0},
      \qquad n\in\N
   \end{align*}
   with
   \begin{align*}\ts
      \lalo_{n+n_0}\, &:=\, \ts
      \max\Bigl\{ |am|,\ 
	 \re\Bigl(
	    \sqrt{\bigl(|k+\frac{1}{2}| -\frac{1}{2} + n_0+ \Aeignumber\bigr)^2  + \Om_-} -|am| \Bigr) 
	 \Bigr\},\\[3ex]
      \laup_{n+n_0}\, &:=\, \ts
      \sqrt{\bigl(|k+\frac{1}{2}| -\frac{1}{2} + n_0+ \Aeignumber\bigr)^2  + \Om_+} + |am|.
   \end{align*}

\end{theorem}
\begin{proof}
   By Lemma~\thmref{lemma:angschur:prerequisites}, the angular operator satisfies conditions \ref{mT1}, \ref{B1}, \ref{B2}, \ref{A1}, \ref{A2}, \ref{D1} and \ref{D2a}, and the domain of the operators $\SC{1}(\la)$ does not depend on $\la$ for $\la\in(|am|,\infty)$.
   Since $B^*$ is surjective, we have 
   \begin{align*}
      \essspec(\SC{1}) = \essspec(\mA)\cap(|am|, \infty) = \emptyset
   \end{align*}
   by Corollary~\thmref{cor:schurvar:essspec}.
   Formula~\eqref{eq:angschur:variaton} now follows from Theorem~\thmref{theorem:schurvar:var2} with $\Dup=|am|$ and $\lae=\infty$.
   Since $D$ is bounded and 
   \begin{align*}
      -|am|\, \|x\|^2\ \le\ \scalar{x}{Dx}\ \le\ |am|\, \|x\|^2,
      \qquad x\in\hil, 
   \end{align*}
   application of Theorem~\thmref{theorem:schurvar:var3} with $\Dcup=\Acup=|am|$ and $\Dclo = \Aclo = -|am|$ yields the estimates~\eqref{eq:angschur:Est}.
   By Theorem~\thmref{theorem:schurvar:var2} and Proposition~\ref{prop:schurvar:n0finite}, 
   $n_0=\min\limits_{\la\in(|am|,\infty)} \dim\specspace_{(-\infty,0)}\SC{1}(\la) < \infty $. 
   Since $(|am|,\la_{m_++1})\subseteq\rho(\SC{1})$, the index shift $n_0$ is constant in this interval.
   Hence also the assertion concerning $n_0$ is proved.
   The explicit two-sided estimates for the eigenvalues $\la_{m_++\Aeignumber}$ are obtained if we insert the estimates~\eqref{eq:sturm} into~\eqref{eq:angschur:Est} and observe that $\la_{m_++\Aeignumber}>|am|$  by definition.
\end{proof}

A result similar to Theorem~\thmref{theorem:angschur:var1}  follows directly from standard perturbation theory (see, e.g., \cite{kato}) applied to the angular operator with $m$ as perturbation parameter.
For convenience, we state this result in the next theorem.
Since with the method from perturbation theory no index shift $n_0$ occurs, a comparison of the results of the following theorem and of Theorem~\thmref{theorem:angschur:var1} leads to a condition for $n_0=0$ (see Propositions~\thmref{prop:angschur:n00} and \thmref{prop:angschur:n00B}).


\begin{theorem}\label{theorem:angschur:SPT}
   Let $\la_{\Aeignumber}$ be the $\Aeignumber$th eigenvalue of the angular operator $\mA$ with the ordering described above. Then for all $\Aeignumber\in\N$ we have
   \begin{align*}\ts
      \la_{\Aeignumber}\ &\ge\
      \re\Bigl(\sqrt{\bigl(|k+\frac{1}{2}| -\frac{1}{2} + \Aeignumber \bigr)^2  + \Om_-}\ \Bigl)\ - |am|\
      \\[2ex]
      \la_{\Aeignumber}\ &\le\
      \sqrt{\bigl(|k+\frac{1}{2}| -\frac{1}{2} + \Aeignumber\bigr)^2  + \Om_+} + |am|.
   \end{align*}
   The functions $\Om_-$ and $\Om_+$ are defined in~\eqref{eq:Omega}.
\end{theorem}

\begin{proof}
   Estimates for the eigenvalues of $BB^*$ are given in~\eqref{eq:sturm}.
   Since $B$ and $B^*$ are invertible, the spectrum of $\mB=\tm{0}{B}{B^*}{0}$ is given by $\pointspec(\mB) = \{\pm\sqrt{\nu_n}\, :\, \nu_n \in\sigma(BB^*)\}$.
   Now, application of analytic perturbation theory to the operators $\mB$ and $\mA$ with $m$ as perturbation parameter yields $\sqrt{\nu_n} - |am| \le \la_n \le \sqrt{\nu_n} + |am|$.
\end{proof}

\begin{remark}\label{remark:laQ}
   In addition to the estimates for the eigenvalues of $\mA$ presented in this paper, there is also another method to derive a lower bound for the modulus of the eigenvalues of $\mA$, see~\cite{thesis}.
   This method making use of sesquilinear forms yields the following bound:
   \begin{align}\label{eq:laQ}
	 |\la_n|\, \ge\, \laQ
   \end{align}
   where 
   \begin{align*}
	 \la_Q\, :=\,
      \begin{cases}
	 \sign(k+\frac{1}{2})(a\om + k + \frac{1}{2}) = |a\om+k+\frac{1}{2}|\qquad
	 &\text{if}\ a\om\in [\, -|k+\frac{1}{2}|,\ |k+\frac{1}{2}, ]\\[2ex]
	 2\sqrt{a\om(k+\frac{1}{2})}
	 &\text{if}\ \sign( |k+\frac{1}{2}| )a\om \ge |k+\frac{1}{2}|.
      \end{cases}
   \end{align*}
   In the case $\sign( |k+\frac{1}{2}| )a\om \ge |k+\frac{1}{2}|$ this methods yields no bound for the eigenvalues.
\end{remark}


\subsection{The index shift $n_0$} 
\label{subsec:indexshift}
The index shift $n_0$ does not depend on the choice of $\la_0\in(|am|, \la_{m_++1})$ but, of course, it depends on the physical parameters $a$, $m$, $\om$ and $k$.

The following lemma gives a sufficient condition for the index shift to be nontrivial.
\begin{lemma}\label{lemma:angschur:n0neq0}
   If there exists an eigenvalue $\mu$ of $\mA$ such that 
   \begin{align}\label{eq:angschur:n0mu}
      2\,|am| - \la_{m_++1}\ <\ \mu\ <\ \la_{m_++1},
   \end{align}
   then we have $n_0 \ge 1$. 
   If in addition $\la_{m_++1}\le 3\, |am|$, then there is at least one eigenvalue of 
   $\mA$ in $[-|am|, |am|\,]$.
\end{lemma}
\begin{proof}
   Recall that $\la_{m_++1}$ is the first eigenvalue of $\mA$ which is greater than $|am|$, hence we have $\mu\le|am|$.
   If we also know $\la_{m_++1}\le 3\, |am|$, then \eqref{eq:angschur:n0mu} shows that 
   $\mu > 2|am| - \la_{m_++1}\ge -|am|$. Hence the eigenvalue $\mu$ of $\mA$ lies in $[-|am|, |am|\,]$.\\
   It remains to be shown that $n_0\ge 1$. 
   Recall that $n_0=\min\limits_{\la>|am|}\dim\specspace_{(-\infty,0)}\SC{1}(\la)$ and that the right hand side is constant on the resolvent set of $\mA$ and non-decreasing with increasing $\la$.
   Hence $n_0 = n(\la) :=\dim\specspace_{(-\infty,0)}\SC{1}(\la)$ for all $\la\in(|am|, \la_{m_++1})$.
%
Let $\mD$ be an arbitrary linear manifold such that $\mD(\SC{1})\subseteq\mD\subseteq\mD(\SCfs{1})$.
In \cite[lemma 2.5]{eschwe} it has been shown that $n_0$  is equal to the dimension of every maximal subspace of 
\begin{align*}
   \mN(\la)\, :=\, \{x\in\mD : \SCfs{1}(\la)[x] < 0 \} \cup \{0\}.
\end{align*}
Since $\mD(\SCfs{1}) = \mD(B^*)$ does not depend on $\la$, we can choose $\mD=\mD(B^*)$.
Therefore it suffices to show that there exists $x\in\mD(B^*)$, $x\neq 0$ and $\widetilde\la\in (|am|, \la_{m_++1})$ such that $\SCfs{1}(\widetilde\la)[x]<0$ because then $x$ spans a onedimensional subspace in $\mN(\widetilde\la)$.
   Since $\mu$ is an eigenvalue of $\mA$, there exists an element $\ttv{x}{y}\in\mD(B^*)\oplus\mD(B)$ such that $(\mA-\mu)\tv{x}{y} =0$, i.e.,
   \begin{align*}
      (-D-\mu)x + By &= 0, &
      B^*x + (D-\mu)y &= 0.
   \end{align*}
   In particular, for $\la\in\rho(D)$ we have $(D-\la)^{-1}B^*x = - (D-\la)^{-1}(D-\mu)y = -y + (\mu-\la)(D-\la)^{-1}y$ and
   $\scalar{B^*x}{y} = \scalar{x}{By}= \scalar{x}{(D+\mu)x}$.
   Thus for every $\la>|am|$
   \begin{align*}
      \SCfs{1}(\la)[x]\ &=\ \bigscalar{x}{(-D-\la)x} - \bigscalar{B^*x}{(D-\la)^{-1}B^*x}\\
      &=\ -\la \ltwo{x}^2 - \scalar{x}{Dx} + \scalar{B^*x}{y} - (\mu-\la)\bigscalar{(D-\la)^{-1}B^*x}{y}\\
      &=\ (\mu-\la)\,\bigl(\ltwo{x}^2 + \ltwo{y}^2 \bigr) - (\mu-\la)^2\,\bigscalar{(D-\la)^{-1}y}{y}.
   \end{align*}
   Since $\la>|am|=\|D\|$, we have $0 < -\bigscalar{(D-\la)^{^-1}y}{y} \le (\la- |am|)^{-1}\ltwo{y}^2$.
   Furthermore, 
it follows from Remark~\thmref{remark:symmetry} 
   that $|x(\theta)| = |y(\pi-\theta)|$ for all $\theta\in(0,\pi)$ which implies $\ltwo{x}=\ltwo{y}$. 
   Thus we have
   \begin{align*}
      \SCfs{1}(\la)[x]\ \le\ \ltwo{x}^2\, (\la-|am|)^{-1} (\mu-\la)(\mu+\la-2|am|).
   \end{align*}
   For $\widetilde\la := \la_{m_++1} - \frac{1}{2}( \mu + \la_{m_++1} -2|am|)$ it follows from \eqref{eq:angschur:n0mu} that $\widetilde\la\in(|am|,\la_{m_++1})$.
   Moreover, we have $\mu-\widetilde\la < 0$ and $\mu + \widetilde\la - 2|am| = \frac{1}{2}(\mu+\la_{m_++1}) -|am| > 0$ by~\eqref{eq:angschur:n0mu}. 
   Hence it follows that\\
   \hspace*{\fill}\parbox[b]{0.9\textwidth}{
	 \begin{align*}
	       \SCfs{1}(\widetilde\la)[x]\ &\le\ \ltwo{x}^2\, (\widetilde\la-|am|)^{-1} (\mu-\widetilde\la)(\mu+\widetilde\la-2|am|)\ <\ 0.
	 \end{align*}
   }\hfill
\end{proof}

Recall that $\nu_n$, $n\in\N$, are the eigenvalues of $BB^*$.

\begin{proposition}\label{prop:angschur:n00}
   \begin{enumerate}
      \item 
      If there exists $j_0\ge 2$ such that 
      \begin{align}\label{eq:angschur:interval}
	 \begin{aligned}
	    \sqrt{\nu_{n_0+j_0}} - \sqrt{\nu_{n_0+j_0-1}}\ >\ 2\,|am|
	    \quad\text{and}\quad
	    \sqrt{\nu_{n_0+j_0+1}} - \sqrt{\nu_{n_0+j_0}}\ >\ 2\,|am|,
	 \end{aligned}
      \end{align}
      then $n_0=m_+$. 
     
     \item
     If $\|B^{*-1}\|^{-1} > 2|am|$, then the angular operator $\mA$ has no eigenvalues in the interval $[-|am|, |am|\,]$ and we have $n_0=0$ and $m_+=0$.

   \end{enumerate}
\end{proposition}
\begin{proof}
   \begin{proofenumerate}
      \item
      From standard perturbation theory we know that
      \begin{align}\label{eq:angschur:StanPerPos}
	 \sign(n)\sqrt{\nu_{|\Aeignumber|}} - |am| \ 
	 \le\ \la_{\Aeignumber}\ \le\ 
	 \sign(n)\sqrt{\nu_{|\Aeignumber|}} + |am|,
	 \qquad \Aeignumber\in\Z\setminus\{0\}.
      \end{align}
      Hence~\eqref{eq:angschur:interval} implies that the angular operator
      $\mA$ has exactly one eigenvalue in the interval $[\,\sqrt{\nu_{n_0+j_0}}-|am|,\ \sqrt{\nu_{n_0+j_0}}+|am|\,] $.
      Since by~\eqref{eq:angschur:interval} and~\eqref{eq:angschur:Est} both $\la_{n_0+j_0}$ and $\la_{m_++j_0}$ lie in this interval, it follows that $n_0=m_+$.

      \item Assume that $\|B^{*-1}\|^{-1}>2|am|$. 
      Then we have $\sqrt{\nu_1}> 2\,|am|$ for the smallest eigenvalue $\nu_1$ of $BB^*$. 
      From the estimate~\eqref{eq:angschur:StanPerPos} we obtain that
      \begin{align*}
	 \la_{-1}\le -\sqrt{\nu_1} + |am| < -|am|.
	 \quad\text{and}\quad
	 \la_1 \ge \sqrt{\nu_1} - |am| > |am| 
      \end{align*}
      Hence the angular operator $\mA$ has no eigenvalues in $[-|am|, |am|\,]$ which implies $m_+=0$.\\
      For $\la>|am|$ define the set $\mN(\la)$ as in the proof of Lemma~\thmref{lemma:angschur:n0neq0}.
      Since $n_0$ is equal to the maximal dimension of subspaces of $\mN(\la)$ for $\la\in(|am|, \la_{m_++1})$, it suffices to show that $\mN(\la) = \{0\}$ for $\la$ close enough to $|am|$.
      To this end fix an arbitrary $x\in\mD(\SCfs{1})=\mD(B^*)$.
      Then it is easy to see that for all $\la>|am|$
      \begin{align*}
	 \SCfs{1}(\la)[x]\ &=\ \bigscalar{x}{(-D-\la)x} - \bigscalar{B^*x}{(D-\la)^{-1}B^*x}\\[0.5ex]
	 &\ge\ (-|am|-\la)\,\ltwo{x}^2 + (|am|+\la)^{-1}\,\ltwo{B^*x}^2\\[0.5ex]
	 &\ge\ (|am|+\la)^{-1}\,\ltwo{x}^2\, 
	 \bigl( \,\|B^{*-1}\|^{-2} - (\la+|am|)^2 \bigr).
      \end{align*}
      Since by assumption $\|B^{*-1}\|^{-1}>2|am|$, we have $\SCfs{1}(\la)[x] > 0$ for all $x\in\mD(B^*)\setminus\{0\}$ if $\la$ is sufficiently close to $|am|$. 
      \qedhere

   \end{proofenumerate}
\end{proof}

The next proposition follows immediately from Proposition~\thmref{prop:angschur:n00}.
Recall that $\lalo_{n+n_0}$ and $\laup_{n+n_0}$ are the upper and lower bounds provided by the variational principle for the $(m_+ + n)$-th eigenvalue of $\mA$ (see Theorem~\thmref{theorem:angschur:var1}).
\begin{proposition}\label{prop:angschur:n00B}
   \begin{enumerate}
      \item 
      If there exists $j_0\ge 2$ such that 
      \begin{align}
	 \lalo_{n_0+j_0} - \laup_{n_0+j_0-1}\ >\ 0
	 \qquad\text{and}\qquad
	 \lalo_{n_0+j_0+1} - \laup_{n_0+j_0}\ >\ 0,
      \end{align}
     then $n_0=m_+$. 
     
     \item
     If $\lalo_{1} > |am|$, then the angular operator $\mA$ has no eigenvalues in the interval $[-|am|, |am|\,]$ and we have $n_0=0$ and $m_+=0$.
   \end{enumerate}
\end{proposition}

Note that we have 
\begin{multline*}
   \lim\limits_{N\rightarrow\infty} \lalo_{N+1} - \laup_N\  \\
   \begin{aligned}
      &=\ -2|am| + 
      \lim\limits_{N\rightarrow\infty} 
      \sqrt{\ts (|k+\frac{1}{2}| + N + \frac{1}{2})^2 + \Om_-}
      -\sqrt{\ts (|k+\frac{1}{2}| + N - \frac{1}{2})^2 + \Om_+}\\
      &=\ -2|am| + 1;
   \end{aligned}
\end{multline*}
therefore, (i) of Proposition~\thmref{prop:angschur:n00B} holds whenever $|am|< \frac{1}{2}$.

It can be shown that for fixed parameters $a$, $m$ and $\om$, we have 
$\lim\limits_{|k|\rightarrow\infty}\|B^{-1}\| = 0$, see~\cite{thesis}.
Hence, if the norm of the wave number $k$ is large enough, then the angular operator has no eigenvalues in $[-|am|,|am|\,]$ and the index shift $n_0$ vanishes.


\subsection{Comparison with numerical values} 
\label{subsec:comparison}
Suffern, Fackerell and Cosgrove~\cite{SFC83} have obtained numerical approximations of the eigenvalues $\la$ of the angular operator by expanding the solution of the angular equation in terms of hyper\-geo\-metric functions. 
They derived a three-term recurrence relation for the coefficients in the series ansatz. 
For the eigenvalue $\la$ they found an expansion with respect to $a(m-\om)$ and $a(m+\om)$ 
\begin{align*}
   \la\, =\, 
   \sum\limits_{r,s} C_{r,s} a^{r+s} (m-\om)^r(m+\om)^s
\end{align*}
with the coefficients $C_{r,s}$ obtained from the recurrence relation for the eigenfunctions.
The numerical values of~\cite{SFC83} will be denoted by $\lanumS[n]$. 
They differ from the eigenvalues given in this work by a factor $-1$ due to the choice of the sign of the constant of separation in the separation ansatz for the Dirac equation.

\par\medskip
For fixed values of $am$ and $a\om$,
Tables~\thmref{table:comparison:la2} and~\thmref{table:comparison:la1} 
contain the numerical values for the first positive and first negative eigenvalues $\lanumS[\pm 1]$ tabulated in \cite{SFC83} and the analytical
lower and upper bounds  $\lalo_1$ and $\laup_1$ from Theorem~\thmref{theorem:angschur:var1} for wave numbers $k=-5,\, \dots,\, 4$.
In addition, the values for the lower bound $\laQ$ from Remark~\thmref{remark:laQ} are listed.
For all physical parameters under consideration, apart from the case $am=0.25$, $a\om=0.75$, $k=-1$, we have
\begin{align*}
   \|B^{-1}\|^{-1}\ =\ \sqrt{\nu_1}\ 
   \ge\ \re \Bigl( \sqrt{\ts (|k+\frac{1}{2}|+\frac{1}{2})^2 + \Om_-}\ \Bigr)\
   >\ 2|am| 
\end{align*}
where $\nu_1$ is the first eigenvalue of $BB^*$, see~\eqref{eq:sturm}, so that we have $n_0=0$ and $m_+=0$ by Proposition~\thmref{prop:angschur:n00}~(ii).
Therefore, the first positive eigenvalue is indeed the analytic continuation of the first positive eigenvalue in the case $a=0$.
The case $am=0.25$, $a\om=0.75$, $k=-1$ is discussed in the subsequent remark.

\begin{remark}[$am=0.25$, $a\om=0.75$, $k=-1$]\
   \begin{enumerate}\label{remark:comparison:discussion}
      
      \item 
      In this case, the bound $\laQ$ from Remark~\thmref{remark:laQ} is not defined because of $\sign(k+\frac{1}{2})a\om = -0.75 < -\frac{1}{2}-|k+\frac{1}{2}|$.

      \item
      The inequalities~\eqref{eq:sturm} yield no positive upper bound for $\|B^{-1}\|$ so that we cannot use Proposition~\thmref{prop:angschur:n00}~(ii) to conclude $n_0=m_+=0$.
      However, since $|am|<\frac{1}{2}$, it follows from Proposition~\thmref{prop:angschur:n00B}~(i) that $n_0=m_+$.
      By Theorem~\thmref{theorem:angschur:SPT} we still have 
      \begin{align}\label{eq:comparison:mPerturbation}
	 -0.25\ \le\ \sqrt{\nu_1} - |am| \le\ \la_1\le \sqrt{\nu_1} + |am|\ \le 1.28078
      \end{align}
      where $\la_1$ is the analytic continuation of the first positive eigenvalue in the case $a=0$
      and $\nu_1$ is the first eigenvalue of $BB^*$ which we have estimated according to~\eqref{eq:sturm}.

      \item
      Even a positive lower bound for $\la_1$ can be obtained by means of analytic perturbation theory if $a$ is treated as the perturbation parameter.
      For $a=0$ we have $\la_n = \sign(n)\bigl( |k+\frac{1}{2}|-\frac{1}{2}+ |n| \bigr) = n$; 
      hence for the given physical parameters we obtain 
      \begin{align}\label{eq:comparison:aPerturbation}
	 n - 0.75\ \le \la_n\ \le\ n + 0.75, \qquad n\in\Z\setminus\{0\}.
      \end{align}
      In particular it follows that $0.25\le\la_1$. 
      For all other values of $n$, however, the bounds $\lalo_n$ and $\laup_n$ obtained from the more elaborate estimates in Theorem~\thmref{theorem:angschur:var1} (where $m$ plays the role of the perturbation parameter) yield tighter bounds than the formula above as can be seen in Table~\ref{table:comparison:higher2-k0,-1}.
   \end{enumerate}
   Combining~\eqref{eq:comparison:mPerturbation} and~\eqref{eq:comparison:aPerturbation} we obtain
   $0.25\ \le\ \la_1\ \le\ 1.28078$.
\end{remark}

\begin{remark}\label{remark:comparison:discussion2}
   In some cases, the bounds can be further improved.
      For $am=0.005$ and $a\om = 0.015$ and $k\in\{-5,\, \dots,\, 4\}$ we have $\spectrum(\mA)\cap [-|am|, |am|\,] =\emptyset$ because of $|\la_{\pm 1}| \ge \lalo_1 = \sqrt{\nu_1} - |am| > |am|$.
      Furthermore, $\|B^{-1}\|^{-1} = \sqrt{\nu_1} > |am|$ so that the assumption of Lemma~3.38 of \cite{thesis} is satisfied.
      Hence it follows:
   \begin{enumerate}
      \item For $k = 0,\, \dots,\, 4$ we have $\bigl(-\|B^{-1}\|^{-1}, -|am|\bigr)\cap\spectrum{(\mA)} = \emptyset$, hence
      \begin{align*}
	 \la_{-1}\ \le\ \|B^{-1}\|^{-1}\ =\ -\laup_1 - |am|. 
      \end{align*}
      
      \item For $k = -5,\, \dots,\, -1$ we have $\bigl(-\|B^{-1}\|^{-1}, -|am|\bigr)\cap\spectrum{(\mA)} = \emptyset$, hence
      \begin{align*}
	 \la_{1}\ \ge\ \|B^{-1}\|^{-1}\ =\ \lalo_1 + |am|. 
      \end{align*}
      
   \end{enumerate}
   Analogously, for $am=0.25$, $a\om=0.75$ the upper bound for $\la_{-1}$ can be improved if $k=0,\,\dots,\,4$ and the lower bound for $\la_{1}$ can be improved if $k=-5,\,\dots,\, -2$.\\
   Note, however, that for $k=-1$ the assumptions of Lemma~3.38 in \cite{thesis} are not fulfilled.
\end{remark}

The discussion in Remarks~\thmref{remark:comparison:discussion} and \thmref{remark:comparison:discussion2} shows that it is not easy to decide a priori which analytic bound gives the sharpest bound for the eigenvalues of $\mA$. 
It seems that often a combination of the various estimates yields the best result.

It can be seen from the tables that in most cases the estimate $\lalo_1$ yields the sharpest lower bound.
On the other hand, Figures~\ref{figure:comparison:varya:k0_m0025_a} and~\ref{figure:comparison:varya:m025_om075_k0} suggest that for increasing $am$ and $a\om$ the estimate $\laQ$ provides a better lower bound for the smallest positive eigenvalue than $\lalo$ does. \\
%

In Figure~\ref{figure:comparison:la1} the numerical values $\lanumplusS[1]$ and $\lanumminusS[1]$ together with the analytic lower bounds $\pm\laQ$ 
and the analytic upper and lower bounds $\lalo_1$ and $\laup_1$ from Theorem~\ref{theorem:angschur:var1} for the lowest eigenvalues as functions of $k$ are plotted. 

\begin{table}[h] 
   \begin{center}
      \begin{tabular}{r|rrSSr}
	 \hline\hline
	 \multicolumn{6}{c}{ $am=0.25$,\quad $a\om =0.75$}\\
	 \hline\hline
	 \multicolumn{6}{c}{}\\
	 & \multicolumn{1}{c}{$\laQ$}&
	 \multicolumn{1}{c}{$\lalo_1$}&
	 \multicolumn{1}{>{\columncolor{SuffernColour}}c}{$\lanumplusS[1]$}& 
	 \multicolumn{1}{>{\columncolor{SuffernColour}}c}{$\lanumminusS[1]$}&
	 \multicolumn{1}{c}{$\laup_1$}\\
	 \hline
	 $k=-5$ 
	 &3.75000   &3.93330   &4.29756 &-4.34936 &4.61606
	 \\                                                        
	 $-4 $                                                     
	 &2.75000   &2.91228   &3.30870 &-3.37371 &3.65037
	 \\                                                        
	 $-3$                                                      
	 &1.75000   &1.87132   &2.32657 &-2.41349 &2.71221
	 \\                                                        
	 $-2 $                                                     
	 &0.75000   &0.75000   &1.35984 &-1.48903 &1.85078
	 \\                                                        
	 $-1$                                                      
	 &undefined&(0.25000) &0.44058 &-0.67315 &(1.28078)
	 \\                                                        
	 $0$                                                       
	 &1.22474   &0.75000   &1.59764 &-1.47645 &1.85078
	 \\                                                        
	 $1$                                                       
	 &2.25000   &2.09521   &2.65654 &-2.57663 &2.90754
	 \\                                                        
	 $2$                                                       
	 &3.25000   &3.21410   &3.68229 &-3.62219 &3.93273
	 \\                                                        
	 $3$                                                       
	 &4.25000   &4.27769   &4.69685 &-4.64856 &4.94707
	 \\                                                        
	 $4$                                                       
	 &5.25000   &5.31776   &5.70622 &-5.66583 &5.95636
	 \\
      \end{tabular}
   \end{center}
      \bigskip
      \caption{%
	 Analytic bounds and numerical approximations for the first positive and first negative 
	 eigenvalue of $\mA$.
	 The estimate $\laQ$ from Remark~\thmref{remark:laQ} is a lower bound for 
	 $|\la_{\pm 1}|$.
	 $\lalo_1$ and $\laup_1$ from Theorem~\thmref{theorem:angschur:var1} are upper and lower 
	 bounds for $\la_{\pm1}$.
	 The values $\lanumplusS[1]$ and $\lanumminusS[1]$ for the first positive and the first
	 negative eigenvalue of $\mA$ are taken from~\cite{SFC83}.
	 We have obtained the numerical values $\lanumplus[1]$ and $\lanumminus[1]$ 
	 by approximating a solution of the continued fraction equation for $\la$.
	 Note that for $k=0,\,\dots,\, 4$ the upper bound for $\la_{-1}$ can be further improved, while for $k=-2,\,\dots,\, -5$ the lower bound for $\la_1$ can be improved, see Remark~\thmref{remark:comparison:discussion2}.
	 For $k=-1$ see the discussion in Remark~\thmref{remark:comparison:discussion}.
\label{table:comparison:la2} 
  }

\end{table} 

\begin{table}[h] 
\begin{center}
      \begin{tabular}{r|rrSSr}
	 \hline\hline
	 \multicolumn{6}{c}{ $am=0.005$,\quad $a\om =0.015$}\\
	 \hline\hline
	 \multicolumn{6}{c}{}\\
	 &\multicolumn{1}{c}{$\laQ$}&
	 \multicolumn{1}{c}{$\lalo_1$}& 
	 \multicolumn{1}{>{\columncolor{SuffernColour}}c}{$\lanumplusS[1]$}& 
	 \multicolumn{1}{>{\columncolor{SuffernColour}}c}{$\lanumminusS[1]$}&
	 \multicolumn{1}{c}{$\laup_1$}
	 \\
	 \hline
	 $k=-5$ 
	 &4.48500 &4.97998 &4.98591 &-4.98682 &4.99299
	 \\                                                    
	 $-4 $                                                 
	 &3.48500 &3.97997 &3.98611 &-3.98723 &3.99373
	 \\                                                    
	 $-3$                                                  
	 &2.48500 &2.97996 &2.98643 &-2.98786 &2.99498 
	 \\                                                    
	 $-2 $                                                 
	 &1.48500 &1.97994 &1.98700 &-1.98901 &1.99749
	 \\                                                    
	 $-1$                                                  
	 &0.48500 &0.97989 &0.98834 &-0.99170 &1.00500
	 \\                                                    
	 $0$                                                   
	 &0.51500 &0.99500 &1.01167 &-1.00836 &1.01989
	 \\                                                    
	 $1$                                                   
	 &1.51500 &2.00249 &2.01300 &-2.01101 &2.01994 
	 \\                                                    
	 $2$                                                   
	 &2.51500 &3.00498 &3.01357 &-3.01215 &3.01996
	 \\                                                    
	 $3$                                                   
	 &3.51500 &4.00623 &4.01389 &-4.01278 &4.01997
	 \\                                                    
	 $4$                                                   
	 &4.51500 &5.00699 &5.01409 &-5.01318 &5.01998
	 \\
      \end{tabular}
\end{center}
      \bigskip
      \caption{\label{table:comparison:la1}%
	    Analytic bounds and numerical approximations for the first positive and first negative 
	    eigenvalue of $\mA$.
	    The estimate $\laQ$ from Remark~\thmref{remark:laQ} is a lower bound for $|\la_{\pm 1}|$.
	    $\lalo_1$ and $\laup_1$ obtained in Theorem~\thmref{theorem:angschur:var1} 
	    are upper and lower bounds for $\la_{\pm 1}$. 
	    The values $\lanumplusS[1]$ and $\lanumminusS[1]$ are the first positive and the 
	    first negative eigenvalue of $\mA$ calculated numerically by 
	    Suffern et~al.~\cite{SFC83}.
	    Note that for $k=0,\,\dots,\, 4$ the upper bound for $\la_{-1}$ can be further improved, while for $k=-1,\,\dots,\, -5$ the lower bound for $\la_1$ can be improved, see Remark~\thmref{remark:comparison:discussion2}.
	 }
\end{table} 

\begin{figure}[h] 
\normalsize
   \begin{center}
      \input{plots/lowest2Qloup.tex}
   \end{center}

   \ \vspace*{\fill}

   \begin{center}
      \input{plots/lowest1Qloup.tex}
      \medskip
      \caption{\label{figure:comparison:la1}%
	 The plots show the lower bound $\laQ$ for the absolute value of the eigenvalues of $\mA$ and the analytic upper and lower bounds $\lalo_1$ and $\laup_1$ from Theorem~\ref{theorem:angschur:var1}
	 for two different values of $(am,\, a\om)$.
	 In addition, the numerical values for the first positive and first negative eigenvalue of $\mA$ from~\cite{SFC83} are plotted.
	 The analytic bounds have not been plotted in the interval $(-1,0)$ because for wave numbers $k$ in that interval the angular operator is not uniquely defined as a selfadjoint operator. 
	 Note that for $(am,a\om)=(0.25,\ 0.75)$ the bound $\laQ$ is not defined for $k\in(-1.25,\, 0)$.
	 In the case $(am, a\om)=(0.005,0.015)$ the analytic lower and upper bounds $\lalo_1$ and $\laup_1$ are so close to each other that they seem to coincide in this resolution.
      }
   \end{center}
\end{figure}


\begin{table}[h] 
      \normalsize
   \begin{center}
      \begin{tabular}{r|rSSr}
	 \hline\hline
	 \multicolumn{5}{c}{ $am=0.25,\quad a\om=0.75$}\\
	 \hline\hline
	 \multicolumn{5}{c}{}\\
	 $k=0$& 
	 \multicolumn{1}{c}{$\lalo_n$}& 
	 \multicolumn{1}{>{\columncolor{SuffernColour}}c}{$\lanumplusS[n]$}& 
	 \multicolumn{1}{>{\columncolor{SuffernColour}}c}{$\lanumminusS[n]$}&
	 \multicolumn{1}{c}{$\laup_n$}
	 \\
	 \hline
	 $n=1$& 0.75000&  1.59764&  -1.47645&  1.85078\\
	 $  2$& 1.75000&  2.22587&  -2.23549&  2.60850\\
	 $  3$& 2.75000&  3.17408&  -3.16265&  3.50000\\
	 $  4$& 3.75000&  4.13127&  -4.12446&  4.44076\\
	 $  5$& 4.75000&  5.10533&  -5.10083&  5.40388\\
	 
	 \multicolumn{5}{c}{}\\ 
	 $k=-1$& 
	 \multicolumn{1}{c}{$\lalo_n$}& 
	 \multicolumn{1}{>{\columncolor{SuffernColour}}c}{$\lanumplusS[n]$}& 
	 \multicolumn{1}{>{\columncolor{SuffernColour}}c}{$\lanumminusS[n]$}&
	 \multicolumn{1}{c}{$\laup_n$}
	 \\
	 \hline
	 $n=1$& (0.25000)& 0.44058&  -0.67315&  (1.28078) \\
	 $  2$& 1.33114&   1.84225&  -1.87948&   2.26556  \\
	 $  3$& 2.48861&   2.90717&  -2.92301&   3.26040  \\
	 $  4$& 3.55789&   3.93475&  -3.94336&   4.25780  \\
	 $  5$& 4.59768&   4.94973&  -4.95513&   5.25625  \\
      \end{tabular}
   \end{center}

      \bigskip
      \caption{\label{table:comparison:higher2-k0,-1}%
	 For $am=0.25$, $a\om=0.75$ and $k=0,\ -1$ the numerical values $\lanumminusS[n]$ and $\lanumplusS[n]$ and the lower and upper bounds $\lalo_n$ and $\laup_n$ from Theorem~\thmref{theorem:angschur:var1} are shown.
	 For $k=-1,\, n=1$ we refer to Remark~\thmref{remark:comparison:discussion}.
      } 
\end{table}


\begin{table}[h] 
   \normalsize
   \begin{center}
      \begin{tabular}{r|rSSr}
	 \hline\hline
	 \multicolumn{5}{c}{ $am=0.005,\quad a\om=0.015$}\\
	 \hline\hline
	 \multicolumn{5}{c}{}\\
	 $k=0$& 
	 \multicolumn{1}{c}{$\lalo_n$}& 
	 \multicolumn{1}{>{\columncolor{SuffernColour}}c}{$\lanumplusS[n]$}& 
	 \multicolumn{1}{>{\columncolor{SuffernColour}}c}{$\lanumminusS[n]$}&
	 \multicolumn{1}{c}{$\laup_n$}
	 \\
	 \hline
	 $n=1$& 0.99500&  1.01167&  -1.00836& 1.01989\\
	 $  2$& 1.99500&  2.00435&  -2.00369& 2.01249\\
	 $  3$& 2.99500&  3.00273&  -3.00245& 3.01000\\
	 $  4$& 3.99500&  4.00180&  -4.00184& 4.00875 \\
	 $  5$& 4.99500&  5.00158&  -5.00148& 5.00800\\

	 \multicolumn{5}{c}{}\\ 
	 $k=-1$& 
	 \multicolumn{1}{c}{$\lalo_n$}& 
	 \multicolumn{1}{>{\columncolor{SuffernColour}}c}{$\lanumplusS[n]$}& 
	 \multicolumn{1}{>{\columncolor{SuffernColour}}c}{$\lanumminusS[n]$}&
	 \multicolumn{1}{c}{$\laup_n$}
	 \\
	 \hline
	 $n=1$&  0.97989&  0.98834&  -0.99170&  1.00500\\
	 $  2$&  1.98749&  1.99567&  -1.99636&  2.00500\\
	 $  3$&  2.99000&  2.99730&  -2.99759&  3.00500\\
	 $  4$&  3.99125&  3.99803&  -3.99819&  4.00500\\
	 $  5$&  4.99200&  4.99845&  -4.99855&  5.00500\\
      \end{tabular}
      \bigskip
      \caption{\label{table:comparison:higher1-k0,-1}%
	 For $am=0.015$, $a\om=0.025$ and $k=0,\ -1$ the numerical values $\lanumminusS[n]$ and $\lanumplusS[n]$ and the lower and upper bounds $\lalo_n$ and $\laup_n$ from Theorem~\thmref{theorem:angschur:var1} are shown.} 
   \end{center}
\end{table}

\begin{figure}[h] 
      \normalsize
   \begin{center}
      \input{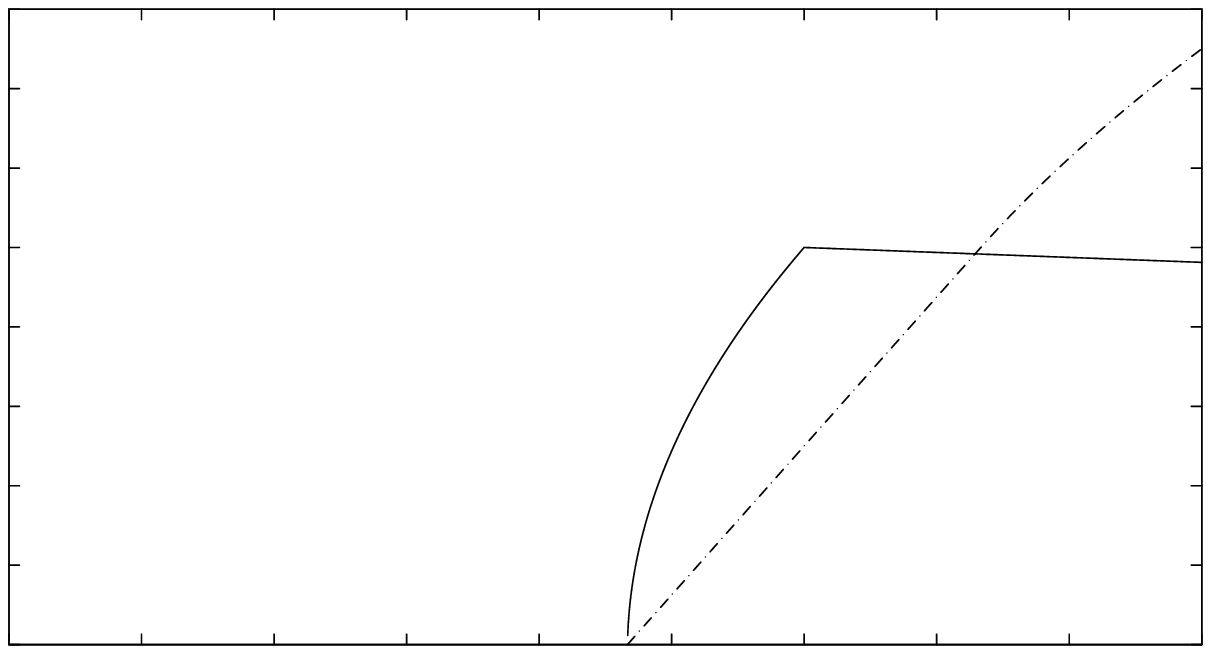}\\[4ex] 	
      \input{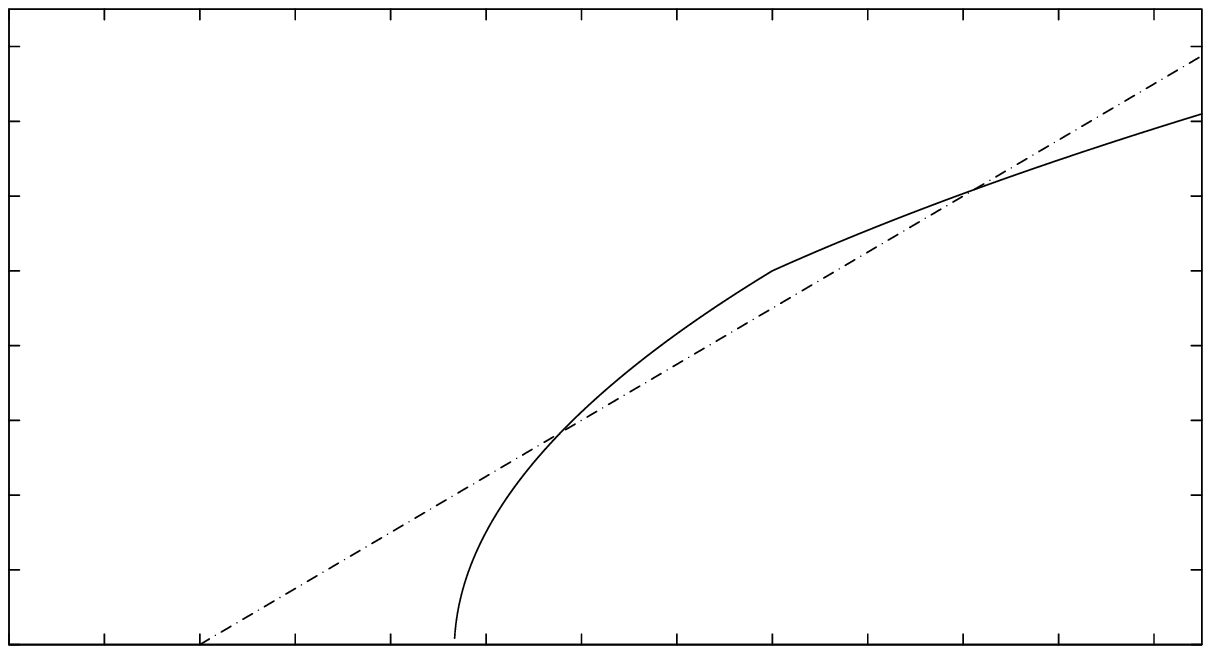}\\ 		
      \input{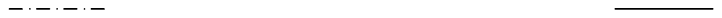}\\	 	
      \medskip
      \caption{\label{figure:comparison:varya:k0_m0025_a}%
	 Lower bounds for the eigenvalue of $\mA$ with smallest modulus for 
	 $k,\ m$ and $\om$ fixed. 
	 Note that for both cases of $k$, for each of the plotted estimates there is an interval for $a$ 
	 where it provides a larger lower bound for the modulus of the eigenvalues of $\mA$ 
	 than the other bound.
      }
   \end{center}
\end{figure} 

\begin{figure}[h] 
      \normalsize
   \begin{center}

      \input{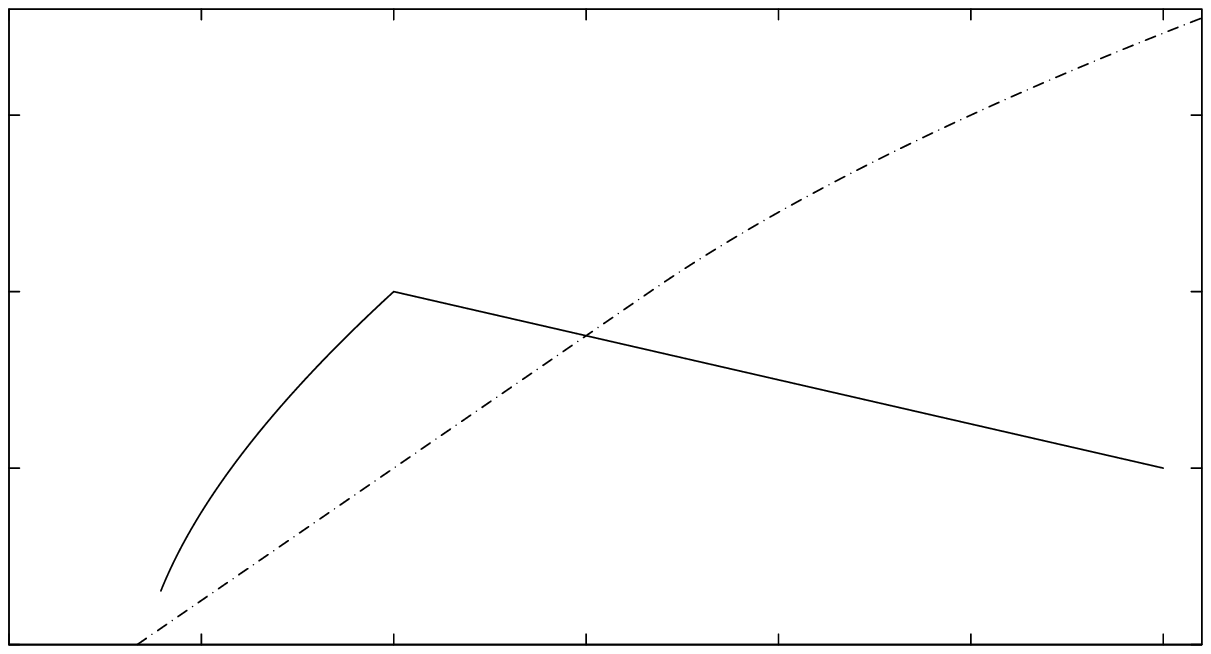}\\[4ex] 	
   \end{center}
   \begin{center}
      \input{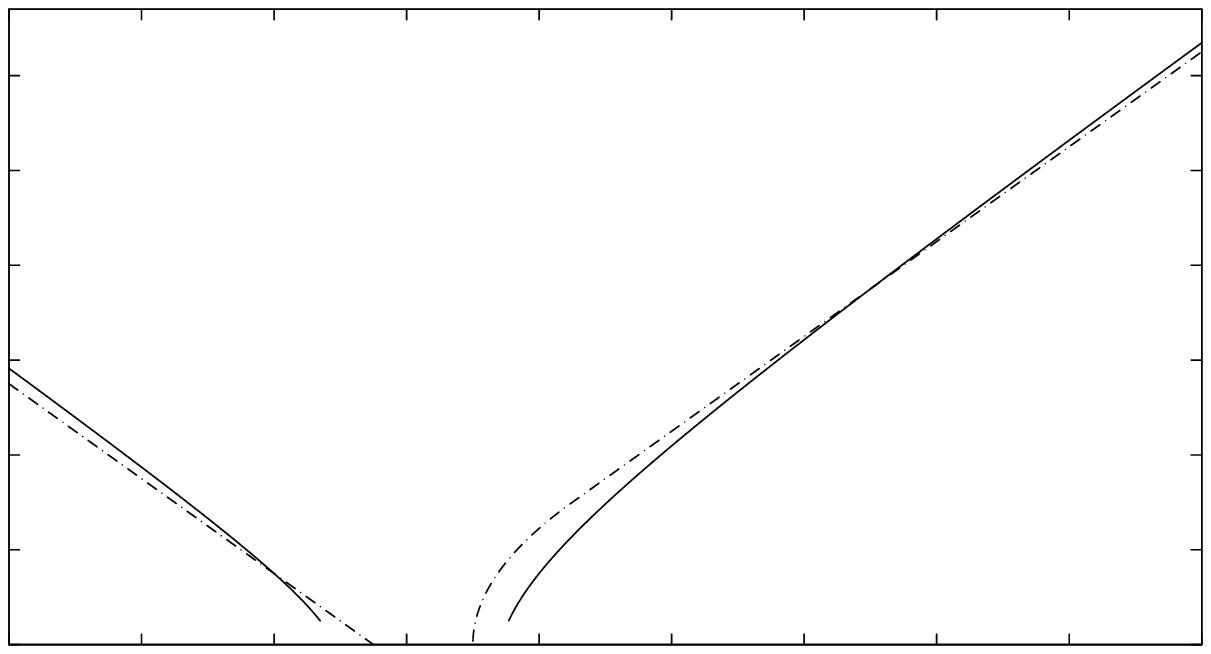}\\ 	
      
      \input{plots/lowerboundslegend}\\ 	
      \medskip
      \caption{\label{figure:comparison:varya:m025_om075_k0}%
	    Lower bounds for the modulus of the eigenvalues $\la$ of $\mA$ for $m=0.25$ and $\om=0.75$ fixed. 
      In the first graph, the bounds are plotted as functions of $a$ for $k=0$ fixed.
      The second graph shows the bounds as functions of $k$ with $a=1$ fixed. Note that physically only the values for integral values of $k$ make sense.
      }
   \end{center}
\end{figure} 


\clearpage
   \hspace*{\fill}
   {\bf Acknowledgement.}
   \hspace*{\fill}\\[1ex]
      The author wishes to express her gratitude to 
      the DFG (German Research Foundation), grant TR 368/6-1 for the financial support.



\begin{thebibliography}{LMMT01} 

\bibitem[ALMS94]{ALMS}
F.~V. Atkinson, H.~Langer, R.~Mennicken, and A.~A. Shkalikov.
\newblock The essential spectrum of some matrix operators.
\newblock {\em Math. Nachr.}, 167:5--20, 1994.

\bibitem[ALMS96]{ALMSau}
V.~Adamyan, H.~Langer, R.~Mennicken, and J.~Saurer.
\newblock Spectral components of selfadjoint block operator matrices with
  unbounded entries.
\newblock {\em Math. Nachr.}, 178:43--80, 1996.

\bibitem[BEL00]{BEL00}
P.~Binding, D.~Eschw{\'e}, and H.~Langer.
\newblock Variational principles for real eigenvalues of self-adjoint operator
  pencils.
\newblock {\em Integral Equations Operator Theory}, 38(2):190--206, 2000.

\bibitem[BSW05]{BSW}
D.~Batic, H.~Schmid, and M.~Winklmeier.
\newblock On the eigenvalues of the {C}handrasekhar-{P}age angular equation.
\newblock {\em J. Math. Phys.}, 46(1):012504--35, 2005.
\newblock {Electronically available: \href{http://lanl.arxiv.org/abs/math-ph/0402047v1}{arXiv:gr-qc/0512112}.}

\bibitem[Cha98]{chandrasekhar}
S.~Chandrasekhar.
\newblock {\em The mathematical theory of black holes}.
\newblock Oxford Classic Texts in the Physical Sciences. The Clarendon Press
  Oxford University Press, New York, 1998.
\newblock Reprint of the 1992 edition.

\bibitem[DES00a]{DES00a}
J.~Dolbeault, M.~J. Esteban, and E.~S{\'e}r{\'e}.
\newblock On the eigenvalues of operators with gaps. {A}pplication to {D}irac
  operators.
\newblock {\em J. Funct. Anal.}, 174(1):208--226, 2000.
\newblock{Electronically available:
\href{http://www.ma.utexas.edu/mp_arc-bin/mpa?yn=99-199}{mp\_arc 99-199}.}

\bibitem[DES00b]{DES00b}
J.~Dolbeault, M.~J. Esteban, and E.~S{\'e}r{\'e}.
\newblock Variational characterization for eigenvalues of {D}irac operators.
\newblock {\em Calc. Var. Partial Differential Equations}, 10(4):321--347,
  2000.
\newblock {Electronically available:
\href{http://www.ma.utexas.edu/mp_arc-bin/mpa?yn=98-177}{mp\_arc 98-177}.
}

\bibitem[EL04]{eschwe}
D.~Eschw{\'e} and M.~Langer.
\newblock Variational principles for eigenvalues of self-adjoint operator
  functions.
\newblock {\em Integral Equations Operator Theory}, 49(3):287--321, 2004.

\bibitem[GLS99]{GLS}
M.~Griesemer, R.~T. Lewis, and H.~Siedentop.
\newblock A minimax principle for eigenvalues in spectral gaps: {D}irac
  operators with {C}oulomb potentials.
\newblock {\em Doc. Math.}, 4:275--283, 1999.
\newblock {Electronically available:
\href{http://www.math.uni-bielefeld.de/documenta/vol-04/10.html}{Doc. Math.}
}

\bibitem[GS99]{GS}
M.~Griesemer and H.~Siedentop.
\newblock A minimax principle for the eigenvalues in spectral gaps.
\newblock {\em J. London Math. Soc. (2)}, 60(2):490--500, 1999.

\bibitem[Kat80]{kato}
T.~Kato.
\newblock {\em Perturbation Theory for Linear Operators}.
\newblock Springer-Verlag, Berlin Heidelberg New York, second edition, 1980.

\bibitem[KLT04]{KLT04}
M.~Kraus, M.~Langer, and C.~Tretter.
\newblock Variational principles and eigenvalue estimates for unbounded block
  operator matrices and applications.
\newblock {\em J. Comput. Appl. Math.}, 171(1-2):311--334, 2004.

\bibitem[LLT02]{LLT02}
H.~Langer, M.~Langer, and C.~Tretter.
\newblock Variational principles for eigenvalues of block operator matrices.
\newblock {\em Indiana Univ. Math. J.}, 51(6):1427--1459, 2002.

\bibitem[LMMT01]{LMMT01}
H.~Langer, A.~S. Markus, V.~Matsaev, and C.~Tretter.
\newblock A new concept for block operator matrices: the quadratic numerical
  range.
\newblock {\em Linear Algebra Appl.}, 330(1-3):89--112, 2001.

\bibitem[LT98]{LT98}
H.~Langer and C.~Tretter.
\newblock Spectral decomposition of some nonselfadjoint block operator
  matrices.
\newblock {\em J. Operator Theory}, 39(2):339--359, 1998.

\bibitem[Nag89]{Nagel89}
R.~Nagel.
\newblock Towards a ``matrix theory'' for unbounded operator matrices.
\newblock {\em Math. Z.}, 201(1):57--68, 1989.

\bibitem[RS78]{reed_simonIV}
M.~Reed and B.~Simon.
\newblock {\em Methods of modern mathematical physics. {IV}. {A}nalysis of
  operators}.
\newblock Academic Press [Harcourt Brace Jovanovich Publishers], New York,
  1978.

\bibitem[Sch04]{Schmid}
H.~Schmid.
\newblock Bound state solutions of the {D}irac equation in the extreme {K}err
  geometry.
\newblock {\em Math. Nachr.}, 274/275:117--129, 2004.
\newblock{Electronically available:
\href{http://lanl.arxiv.org/abs/math-ph/0207039v2}{arXiv:0207039v2}.}

\bibitem[SFC83]{SFC83}
K.~G. Suffern, E.~D. Fackerell, and C.~M. Cosgrove.
\newblock Eigenvalues of the {C}handrasekhar-{P}age angular functions.
\newblock {\em J. Math. Phys.}, 24(5):1350--1358, 1983.

\bibitem[Tre00]{Tr99}
C.~Tretter.
\newblock Spectral issues for block operator matrices.
\newblock In {\em Differential equations and mathematical physics (Birmingham,
  AL, 1999)}, volume~16 of {\em AMS/IP Stud. Adv. Math.}, pages 407--423. Amer.
  Math. Soc., Providence, RI, 2000.

\bibitem[Win06]{thesis}
M.~Winklmeier.
\newblock {\em The {D}irac {E}quation in the {K}err-{N}ewman {M}etric -
  {E}stimates for the {E}igenvalues}.
\newblock PhD thesis, Universit\"at Bremen, {D}r. {H}ut M\"unchen, 2006.

\bibitem[WY06]{WY06}
M.~Winklmeier and O.~Yamada.
\newblock Spectral analysis of radial {D}irac operators in the {K}err-{N}ewman
  metric and its applications to time-periodic solutions.
\newblock {\em J. Math. Phys.}, 47(10):102503, 17, 2006.
\newblock{Electronically available:
\href{http://lanl.arxiv.org/abs/gr-qc/0605146v2}{arXiv:0605146v2}.}

\end{thebibliography}


\end{document}